\documentclass[preprint]{vgtc}                     % final (journal style)
\onlineid{0}

\vgtccategory{Research}

\vgtcpapertype{please specify}

\title{Automated Responsive Thematic Mapping with Layout Guides}

\author{%
    \authororcid{Arjen Simons}{0009-0008-1271-180X},
    \authororcid{Sarah Schöttler}{0000-0002-4898-2619},
    \authororcid{Wouter Meulemans}{0000-0002-4978-3400},
    \authororcid{Kevin Verbeek}{0000-0003-3052-4844}, and
    \authororcid{Bettina Speckmann}{0000-0002-8514-7858}
}

\authorfooter{
  \item
    All authors are with Eindhoven University of Technology. 
    E-mail: \{~a.simons1 | s.j.schottler | w.meulemans | k.a.b.verbeek | b.speckmann~\}@tue.nl
}

\abstract{%
Thematic maps visually communicate statistical information about spatial units such as countries or states. They must balance the individual readability of those map elements that carry the statistical information and the overall cartographic context. Nowadays, most maps are not static images, but must flexibly \emph{respond} to a range of device types and display sizes.
Current approaches to responsive thematic mapping are limited: they are labor-intensive for practitioners and often rely on combining disjointed visual encodings to cover different device types. 
In this paper we introduce the first algorithmic framework to efficiently compute responsive thematic maps that smoothly adapt to different display sizes. A key component of our framework is the \emph{layout guide}: a combinatorial structure which encodes the two essential aspects of a thematic map. The first aspect are the visual requirements of each statistical map element (at least their desired width and height), the second aspect is the cartographic context in the form of relative positions of map elements. Our main algorithmic contribution is the \emph{map arranger} which takes a visual container as input and returns a suitable layout guide. The map arranger does so in a \emph{stable} and \emph{consistent} manner: if the container changes only a little, then so does the layout guide, and the same input container always results in the same layout guide. To use our framework, one needs three ingredients: $(1)$ a reference layout, which corresponds to the ``ideal'' thematic map, $(2)$ a total vertical and horizontal order for all map elements (the desired layouts for containers with extreme aspect ratios), and $(3)$ a thematic mapping algorithm that can construct a thematic map from a layout guide. We demonstrate our framework on two types of thematic maps, namely rectangular and Demers cartograms.
}


\teaser{
  \centering
  \includegraphics[alt={Pipeline}]{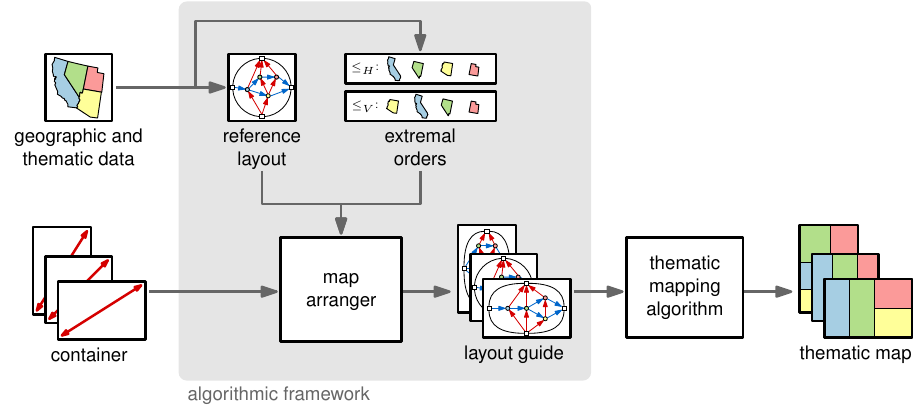}
  \caption{%
    Our algorithmic framework for responsive thematic mapping. To initialize the framework, the user provides a reference layout and extremal orders based on geographic and thematic map data. The map arranger adapts the reference layout to generate layout guides for containers of arbitrary aspect ratios. The layout guides in turn serve as input for a suitable thematic mapping algorithm which generates the final thematic map for each container.
  }
  \label{fig:teaser}
}

\graphicspath{{figs/}{figures/}{pictures/}{images/}{./}} % where to search for the images

\usepackage{tabu}                      % only used for the table example
\usepackage{booktabs}                  % only used for the table example
\usepackage{lipsum}                    % used to generate placeholder text
\usepackage{mwe}                       % used to generate placeholder figures

\usepackage{apxproof}

\usepackage{amsmath}
\usepackage{thm-restate}
\newtheorem{theorem}{Theorem}

\newcommand{\mypar}[1]{\smallskip\noindent\textbf{#1}}

\usepackage{xcolor}
\usepackage{stfloats}
\usepackage{subcaption}
\usepackage{multirow}
\usepackage{wrapfig}

\newcommand{\layout}{\mathcal{L}}
\newcommand{\refLayout}{\mathcal{R}}
\newcommand{\container}{\mathcal{C}}

\begin{document}

%%%%%%%%%%%%%%%%%%%%%%%%%%%%%%%%%%%%%%%%%%%%%%%%%%%%%%%%%%%%%%%%
%%%%%%%%%%%%%%%%%%%%%% START OF THE PAPER %%%%%%%%%%%%%%%%%%%%%%
%%%%%%%%%%%%%%%%%%%%%%%%%%%%%%%%%%%%%%%%%%%%%%%%%%%%%%%%%%%%%%%%

\firstsection{Introduction}

\maketitle

Thematic maps visually communicate geographic information in a way that allows users to recognize spatial patterns, distributions, and trends. To do so, they utilize color (choropleth maps), visual overlays (proportional symbol, necklace~\cite{speckmannNecklaceMaps2010}, or flow maps), deformed map elements (schematic maps and cartograms), or a combination~\cite{dentCartographyThematicMap2009,slocumThematicCartographyGeovisualization2022}. Nowadays, thematic maps are only rarely found as printed images in a physical atlas; instead, they are displayed on a plethora of devices with different display sizes and resolutions. Thematic maps must hence become \emph{responsive}, that is, adapt to the display space they are afforded. Consider, for example, the cartograms of election results that frequently accompany news coverage. These maps should be both readable and informative, independent of the device a reader uses to access the news.

\mypar{Responsive Visualization}
has received significant interest in recent years and there is a range of established approaches to facilitate responsive visualization design. 
Collections of design patterns \cite{kimDesignPatternsTradeOffs2021,hoffswellTechniquesFlexibleResponsive2020} describe general techniques and strategies to adapt visualizations across different devices. These strategies often involve rearranging the layout of the visualization, for example by rescaling or transposing axes, serializing layouts, or relocating individual visual elements. As such they are difficult to directly apply to maps. Libraries and authoring tools support visualization designers with manually designing visualizations for a range of devices \cite{hoffswellTechniquesFlexibleResponsive2020,kimCiceroDeclarativeGrammar2022,kimDupoMixedInitiativeAuthoring2024,andrewsRespVisD3Extension2023,schottlerConstraintBasedBreakpointsResponsive2025}. And for a limited range of visualization techniques, automated approaches have been proposed: a reinforcement learning approach can resize visualizations in Cartesian coordinate systems such as bar and line charts~\cite{wuMobileVisFixerTailoringWeb2021}; multiple views can be rearranged with a framework using simulated annealing~\cite{zengSemiAutomaticLayoutAdaptation2024}; modified seam carving algorithms can resize word clouds and similar visualizations~\cite{wuViSizerVisualizationResizing2013} as well as network visualizations~\cite{digiacomoNetworkVisualizationRetargeting2015}. None of these approaches directly translate to maps.

Responsive visualization requires that visualizations contain internal logic \cite{andrewsResponsiveVisualisation2018} or rules \cite{schottlerConstraintBasedBreakpointsResponsive2025} to adapt the visualization to a range of factors~\cite{horakResponsiveVisualizationDesign2021}. Many different factors can be considered, including, for example, different interaction modalities, usage scenarios, and environments. However, practitioners generally identify the changing \emph{area} and \emph{aspect ratio} of the available display space as the central challenge, particularly in the context of geographic visualizations and, specifically, thematic maps~\cite{schottlerPracticesStrategiesResponsive2025,hoffswellTechniquesFlexibleResponsive2020,houtmanTechniquesChallengesOpportunities2026}. Consequently, the need for further research in this area, with a particular focus on the development of novel techniques, has been highlighted in both the cartographic and visualization research communities~\cite{rothMakingMapsVisualizations2024,houtmanTechniquesChallengesOpportunities2026,schottlerPracticesStrategiesResponsive2025}.

Recent research has found that practitioners often combine different visualizations to cover a range of display sizes for thematic maps~\cite{schottlerPracticesStrategiesResponsive2025}. Responsive visualizations created in this manner can rely on a wide range of existing algorithms and design guidance and do not introduce novel visualization designs that may be unfamiliar to users. However, they also have low visual consistency across devices, may have significant loss of information at smaller sizes, and are time-intensive to develop and maintain for practitioners~\cite{horakResponsiveVisualizationDesign2021}. 
It is also far from obvious how to create effective combinations of visualization types~\cite{schottlerConstraintBasedBreakpointsResponsive2025} and many open questions in this direction remain.

In contrast, our goal in this paper is to develop an algorithmic framework to efficiently compute responsive
thematic maps that smoothly adapt to different display sizes. Below we discuss the cartographic principles, design choices, and trade-offs that underlie our algorithmic model for responsive thematic mapping.

\mypar{Responsive Thematic Mapping.}
Thematic maps visually communicate statistical information about spatial units such as countries and states. This statistical information is carried by specific map elements: the colors or size of map regions in a choropleth map or a cartogram and the size of thematic overlays such as symbols in a proportional symbol, necklace, or flow map. To communicate the spatial context of the statistical information, these \emph{statistical map elements} are arranged in a geographically meaningful manner. For thematic overlays this usually implies some spatial relation of the statistical elements to the regions they represent: either attached to a reference point, or inside a region, or at least in the correct spatial direction in the case of a necklace bead \cite{speckmannNecklaceMaps2010}. For choropleth maps the regions themselves carry the statistical information in the form of color and for cartograms the relative positions of the scaled and abstracted regions attempt to capture those of the regions in the input map.

Thematic maps must always strike a balance between the readability and interpretability of the statistical map elements and the overall cartographic context of the map, as represented by the spatial relations of the statistical map elements to each other or to the base map. Cartographic textbooks~\cite{dentCartographyThematicMap2009,slocumThematicCartographyGeovisualization2022} provide guidance for this trade-off. However, they mostly consider static maps where the aspect ratio of the map can be chosen by the map designer and the resolution is high. In contrast, responsive thematic mapping requires a map design that adapts to an externally imposed \emph{container} with a fixed aspect ratio and resolution. 

To facilitate an explicit trade-off between \emph{readability} and \emph{relative position} of the statistical map elements for a given container, we must first formalize the requirements for responsive thematic mapping.  First of all, there is a choice between always displaying all statistical map elements or reducing the information available to the user. Responsive reference maps follow the second paradigm and use established approaches such as panning, zooming, and map generalization~\cite{rothTypologyOperatorsMaintaining2011,raposoChangeThemeRole2020} to help the user navigate the maps. The corresponding map applications and navigation apps are well-understood and often work reliably across a wide range of devices~\cite{muehlenhausWebCartographyMap2013}.
Zooming and panning can be applied to some thematic map use cases, but practitioners often try to avoid parts of the map being off-screen if they want users to be able to recognize spatial distributions, patterns, and trends and make comparisons across regions
\cite{schottlerPracticesStrategiesResponsive2025}. 
Likewise, visual analytics tools often aggregate thematic overlays into generalized map views and support stable zooming. However, this form of aggregation is often not appropriate for small or medium-sized data sets such as the aforementioned election results.
% We therefore aim to develop thematic mapping techniques that can support these use cases and display the entire, unaggregated, dataset within the available display space.
% Since the purpose of thematic maps is to communicate statistical data with a geographic reference, practitioners argue that one should be able to display the statistical data in its entirety \cite{schottlerPracticesStrategiesResponsive2025}. Zooming and panning can obfuscate spatial patterns and generally hinders accurate comparisons across regions. 
 % Furthermore, these forms of interaction modalities do not easily carry over to mobile devices with small screens and interaction by touch. 
We hence choose to always {\bfseries display the statistical map elements in their entirety} on our responsive thematic maps.

The question now arises how to arrange the statistical map elements in a given container. For the purpose of this paper, we will {\bfseries focus solely on the aspect ratio of the container} and assume that the resolution is sufficient to support our chosen drawing style. When drawing the statistical map elements inside the container, we have to make further decisions. How much can each map element be distorted? In most thematic maps, the area of the map elements is used to communicate the statistical information. Furthermore, if the statistical map elements are regions, then preserving their general shape and aspect ratio will increase the readability of the map. Hence we generally aim to {\bfseries preserve the (relative) area and aspect ratio of the statistical map elements}.

Our second set of design choices centers on the relative position of the statistical map elements. To perfectly preserve relative positions, one must scale the thematic map to fit into the given container. For containers with large aspect ratios, this will invariably result in maps which are too small to be readable and which do not fill the container, thereby wasting valuable display space. 
Practitioners have repeatedly suggested to address this particular challenge by rearranging map layouts, that is, through splitting or segmenting the map and rearranging individual map partitions or spatial units \cite{schottlerPracticesStrategiesResponsive2025,houtmanTechniquesChallengesOpportunities2026}. A user study with a manually rearranged world map showed that this approach can be effective and well-received by users \cite{oeschThematicWorldMaps2025}. Furthermore, the use of cartograms and other schematic map types has been suggested, since their inherent spatial distortion makes them suitable for a wider range of display sizes through more flexible layouts \cite{schottlerPracticesStrategiesResponsive2025}.  

Clearly such approaches do not maintain all relative positions of the ``perfect'' thematic map. However, they {\bfseries choose particular sets of spatial relations to maintain while gradually dropping others}. To make such an approach amenable to algorithmic automation, we need a way to determine which relative positions to drop for a particular aspect ratio. This change in relative position requirements should take place in a {\bfseries stable} manner: if the container changes only a little, then most relative positions should be retained, making it easier for the user to track regions. Furthermore, the change in relative positions should be {\bfseries consistent}: given a (container) aspect ratio, the thematic map should always be the same. In other words, the relative positions in the final map should depend only on the container's aspect ratio, not on the sequence of prior aspect ratio changes. Last but not least, the visualization designer should have some \textbf{control} over the way that the relative positions are adapted as the aspect ratio changes.

\mypar{Contributions.}
In this paper we introduce the first algorithmic framework to efficiently compute responsive thematic maps that smoothly adapt to different display sizes. Central to our algorithmic framework is the concept of a \emph{layout guide} (see Section~\ref{sec:layout-guide}). A layout guide is a combinatorial structure which encodes the two essential aspects of a thematic map: the visual requirements of each statistical map element (at least their desired size and aspect ratio) and the cartographic context in the form of relative positions of map elements. Our main algorithmic contribution is the \emph{map arranger} (see Section~\ref{sec:map-arranger}) which takes the available display space as input and returns a suitable layout guide for this specific container. The map arranger does so in a stable and consistent manner: if the container changes only a little, then so does the layout guide, and the same input container always results in the same layout guide. To use our framework, one needs three ingredients: $(1)$ a reference layout, which corresponds to the thematic map displayed at its natural aspect ratio, $(2)$ a total vertical and horizontal order for all statistical map elements (the desired layouts for containers with extreme aspect ratios), and $(3)$ a thematic mapping algorithm that can compute a thematic map from a layout guide. Our framework is applicable to a wide range of thematic map types; we demonstrate its application to rectangular and Demers cartograms (see Section~\ref{sec:instantiations}). We close in Section~\ref{sec:discussion} with a discussion and future work.

\section{Layout Guide}
\label{sec:layout-guide}

The purpose of a layout guide is to minimally represent a layout of the given map elements, such that we can efficiently determine the required width and height of the map (or an estimation thereof) when drawn according to the represented layout, without actually constructing the map. Furthermore, it should be easy to modify a layout guide to fit a given container.

We represent a layout guide as a graph $\layout = (V, E)$ that is both directed and labeled. The nodes of $\layout$ correspond to the map elements to be visualized and contain some attributes with respect to their drawing requirements. Specifically, we require that each node $v \in V$ stores at least its desired width $w(v)$ and height $h(v)$ in the final visualization. Every directed edge $e = (u,v)$ should capture the relative position between the two map elements corresponding to $u$ and $v$. Here, we choose to use a simple binary label $r(e) \in \{H, V\}$ that indicates whether the relative positioning is either horizontal (H) or vertical (V). Specifically, for a directed edge $e = (u,v)$, $r(e) = H$ means that the map element corresponding to $v$ should be drawn to the right of the map element corresponding to $u$. Similarly, $r(e) = V$ means that the map element corresponding to $v$ should be drawn above the map element corresponding to $u$. We visually indicate these labels by coloring an edge $e$ blue if $r(e) = H$ and red if $r(e) = V$. Furthermore, in the following we refer to edges $e$ with $r(e) = H$ as \emph{horizontal} or \emph{blue} edges, and to edges $e$ with $r(e) = V$ as \emph{vertical} or \emph{red} edges, depending on the context.

To ensure that the layout guide can determine whether a particular layout of the map elements can fit the desired container, we impose additional structure on the underlying graph. Here we use the existing concept of \emph{transversal structures} as introduced by Fusy~\cite{fusy2009transversal}, or equivalently, \emph{rectangular edge labelings} as introduced by Kant and He~\cite{kant1997regular}. Specifically, we add four boundary nodes N (North), E (East), S (South), and W (West) to the layout guide that essentially correspond to the top, right, bottom, and left side of the container, respectively. Note that these nodes do not correspond to actual map elements and hence do not require any attributes. In addition, the edges between two boundary nodes do not have a direction or a label. We then require that the edges of the layout guide form a triangulation where the four boundary nodes form the outer boundary of the triangulation (see \cref{fig:transfersal_structure_example}). In particular, this also means that the layout guide has a clearly defined combinatorial embedding, that is, the order of edges around every node is specified. The edge labels of the layout guide then form a transversal edge-partition if the following two properties hold:
\begin{enumerate}
    \item In clockwise order around every non-boundary node of $\layout$, we encounter the following edge types (at least one each): outgoing vertical (OV) edges, outgoing horizontal (OH) edges, incoming vertical (IV) edges, and incoming horizontal (IH) edges.
    \item All edges incident to the boundary nodes N, E, S, and W (excluding between two boundary nodes) are, respectively, incoming vertical, incoming horizontal, outgoing vertical, and outgoing horizontal.
\end{enumerate}
Such a transversal structure exists only if the triangulation does not contain separating triangles~\cite{fusy2009transversal,kant1997regular}. Such triangulations are also referred to as \emph{irreducible triangulations}. We therefore require the edges of the layout guide to form an irreducible triangulation; this requirement can always be satisfied~\cite{bkk-traingulating-97}.

\begin{figure}[t]
    \centering
    \includegraphics{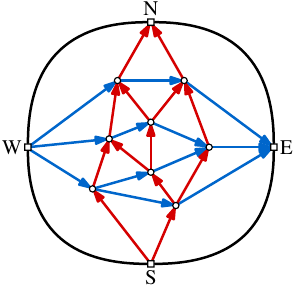}
    \caption{
    Transversal structure on an irreducible triangulation: directed edges are indicated via arrows, and the labels vertical $V$ and horizontal $H$ are indicated using red and blue colors, respectively. Edges between the boundary nodes, W (West), E (East), S (South), and N (North), do not have a direction or label.}
    \label{fig:transfersal_structure_example}
\end{figure}

As a result, the layout guide prescribes the relative positions between map elements: for every map element, the layout guide provides information on which elements should be above, to the right, below, and to the left of it. This allows us to determine the desired width and height of the full map, according to a given layout guide. Let $E_H = \{e \in E\mid r(e) = H\}$ be the horizontal edges and $E_V = \{e \in E\mid r(e) = V\}$ the vertical edges. Both sets of edges form a so-called \emph{bipolar orientation}: a directed acyclic graph with exactly one source (W and S for $E_H$ and $E_V$, respectively) and exactly one sink (E and N for $E_H$ and $E_V$, respectively). Now let $P^{\rightarrow}$ be the set of directed paths from W to E in $E_H$ and let $P^{\uparrow}$ be the set of directed paths from S to N in $E_V$. We can define the desired width and height of a layout guide $\layout$ as follows:
\begin{align}
w(\layout) &= \max_{p \in P^{\rightarrow}} \sum_{v \in p} w(v) \label{eq:widthlayout}\\
h(\layout) &= \max_{p \in P^{\uparrow}} \sum_{v \in p} h(v) \label{eq:heightlayout}
\end{align}
In words, the desired width of a layout guide is the maximum sum of the widths of map elements that are supposed to be arranged horizontally. Similarly, the desired height of a layout guide is the maximum sum of the heights of map elements that are supposed to be arranged vertically. Using $w(\layout)$ and $h(\layout)$ we can easily check if a layout guide is suitable for a given container, as $w(\layout)$ and $h(\layout)$ should not exceed the width and height of the container, respectively. If $w(\layout)$ ($h(\layout)$) exceeds the width (height) of the container, then we refer to a path $p^* \in P^{\rightarrow}$ ($p^* \in P^{\uparrow}$) that induces $w(\layout)$ ($h(\layout)$) as a \emph{critical path}, and to the edges of $p^*$ as \emph{critical edges}. The critical paths indicate where the layout guide should be adapted to ensure that it fits a given container. In the following, we describe how we can compute a suitable layout guide for a given container.

\section{The Map Arranger}
\label{sec:map-arranger}

The core of our algorithmic framework is the \emph{map arranger}. The purpose of the map arranger is to compute a suitable layout guide for any given container. To offer the map designer some control over the layout of map elements and to ensure stability and consistency, the map arranger should be initialized with some information on the visual requirements and cartographic context of the map elements. Therefore, the map arranger is technically a data structure that supports the single operation that takes a container (specifically, its dimensions) as input and returns a suitable layout guide for the given container. In this section, we first describe how the map arranger is initialized and then describe how a layout guide is constructed for a given container. 

\mypar{Initialization.}
To initialize the map arranger, we require a reference layout $\refLayout$ and two extremal total orders $\leq_H$ and $\leq_V$ as input. The reference layout $\refLayout$ takes the form of a layout guide for which there are no constraints due to the container. Thus, the labeled edges of $\refLayout$ capture the relative positions of the map elements in an ideal scenario without space constraints for the resulting map. This layout will function as the starting point for the construction of a layout guide for a given container. The only difference between $\refLayout$ and a layout guide is that $\refLayout$ typically does not explicitly store the desired width and height for each map element (as vertex attributes), but instead stores the desired shape (usually in terms of relative area and aspect ratio). This allows us to adapt the width and height requirements of map elements to the area of the container. 

The extremal orders $\leq_H$ and $\leq_V$ capture the desired order of the map elements horizontally and vertically, respectively, when the aspect ratio of the container is so extreme that the map elements must be arranged linearly. Note that $\leq_H$ and $\leq_V$ can simply be provided as a permutation of the map elements in our implementation, but we use the mathematical notation here for ease of exposition. These orders are used to ensure consistency among the layout guides computed for the different containers: they steer the choices made when layout guides are modified. 
%Therefore, we additionally require that the reference layout $\refLayout$ is consistent with the given extremal orders. Specifically, if $e = (u, v)$ is a (directed) edge of $\refLayout$, then $u \leq_{r(e)} v$ must hold. We will also require the same property for all layout guides we construct. 

Our algorithmic framework does not specify how the reference layout and extremal orders are provided, nor whether it is computed automatically or constructed manually. Obtaining suitable extremal orders should be straightforward, as it simply amounts to ordering the map elements, but obtaining a valid reference layout is less straightforward. If the map elements correspond to geographic regions, then the dual graph (connecting two map elements if the corresponding regions share a border) is a good starting point. However, the resulting graph may contain separating triangles. In that case, we first break the separating triangles by removing some edges. Afterwards, we can add the boundary nodes N, E, S, and W and use the result by Biedl \emph{et al.}~\cite{bkk-traingulating-97} to triangulate the resulting graph without introducing separating triangles. The obtained triangulation is irreducible, and hence there must exist a valid transversal edge-partition for this triangulation, which can be computed efficiently~\cite{kant1997regular,fusy2009transversal}. 
%However, this transversal edge-partition may not be compatible with all possible extremal orders. We therefore recommend to first construct a reference layout $\refLayout$. The bipolar orientations of $\refLayout$ induce horizontal and vertical partial orders. The extremal orders $\leq_H$ and $\leq_V$ must then be chosen to extend these partial orders.

\mypar{Computing a layout guide.}
Now assume we are given a container $\container$ as input to the map arranger. Let $w(\container)$ and $h(\container)$ be the width and height of $\container$, and let $A(\container)$ be its area. We start by initializing the layout guide $\layout$ as the reference layout $\refLayout$. However, instead of the shape requirements stored in $\refLayout$ for each map element, we first compute the desired width and height of each map element based on the given container $\container$, and store that in $\layout$. In the typical scenario where we store the relative area and aspect ratio in $\refLayout$ for each map element, we can simply multiply the relative area by $A(\container)$ to obtain the desired area of the map element, and then extract the width and height via the aspect ratio. However, other computations can be defined here, as long as they result in a desired width and height for each map element.

Next, we can compute the desired width $w(\layout)$ and height $h(\layout)$ of the layout guide $\layout$, as defined in Equations~\ref{eq:widthlayout} and~\ref{eq:heightlayout}. As the bipolar orientations of $\layout$ are directed acyclic graphs, these values can easily be computed in $O(|V| + |E|)$ time using standard algorithms for computing the longest path in directed acyclic graphs (by processing the nodes in topological order). As the layout guide is a triangulation, the running time is also linear in the number of map elements. 

We then compare the aspect ratio of $\layout$ with the aspect ratio of the container $\container$ and determine if we should decrease either the width or height of $\layout$. In the following, we assume without loss of generality that we need to decrease the height of $\layout$, as decreasing the width is symmetric. To decrease the height of $\layout$ we repeatedly choose a critical edge $e^*$ on one of the vertical critical paths of $\layout$. In Section~\ref{sec:eliminate-edge} we describe how we can eliminate this critical edge by either changing its label to $H$ or removing $e^*$ from $\layout$. We repeat this process until the height $h(\layout)$ falls below $h(\container)$, or if the horizontal edges of $\layout$ induce a linear order of the map elements. The resulting layout guide $\layout$ is then returned as output.

For the sake of consistency/stability of layout guides, it is important that the critical edge to be eliminated is always chosen deterministically. In Section~\ref{sec:critical-edge-heuristics} we discuss several different heuristics for choosing the critical edge. Then, by resolving critical edges one by one, always starting from the reference layout, the layout guides of containers with a similar aspect ratio always differ by only a few critical edge elimination operations, which are designed to make only local changes to the layout guide. This ensures that the layout guides are consistent and stable. However, this approach comes at the cost of efficiency, as we start from $\refLayout$ for every given container. To remedy this, observe that the layout guide depends only on the aspect ratio of $\container$. We can thus simply store all intermediate layout guides after every critical edge elimination between the reference layout and the extremal order, which can be precomputed. The map arranger then simply needs to find the correct layout guide based on a binary search on the aspect ratio of $\container$. Although storing all intermediate layout guides would require significantly more memory for the map arranger, this amount of memory can be reduced considerably by using something akin to a version control system. As this space-time tradeoff is not the focus of this paper, we omit the details of such an implementation.

\subsection{Eliminating Critical Edges}\label{sec:eliminate-edge}
Let $e^*$ be the critical edge that we aim to eliminate, and assume without loss of generality that $r(e^*) = V$. To achieve this, we need to change the label of $e^*$ to $H$ or remove $e^*$ from the layout guide $\layout$. Although we can easily change the label, the main challenge lies in re-establishing the properties of a transversal edge-partition, which is required for the layout guide. We modify $\layout$ using the following three operations on an edge $e = (u,v)$ (see~\cref{fig:operations}):
\begin{description}
    \item[Relabel:] Change the label $r(e)$ from $H$ to $V$ or vice versa;
    \item[Redirect:] Replace $e$ with the edge $e' = (v,u)$ where $r(e') = r(e)$;
    \item[Diagonal flip:] Since $\layout$ is a triangulation, $e$ is adjacent to exactly two triangles $(x, u, v)$ and $(y, v, u)$. Replace $e$ by the edge $e' = (x, y)$ or $e'' = (y, x)$ where $r(e') = r(e'') = r(e)$. Based on the change in direction from $e$ to either $e'$ or $e''$, we call the diagonal flip clockwise (CW) or counterclockwise (CCW).
\end{description}
\begin{figure}
    \centering
    \includegraphics[page=2]{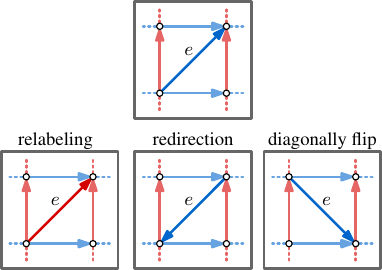}
    \caption{
    Three operations applied to edge $e$, shown at the top. From left to right: relabeling, redirection (which results in an invalid transversal edge-partition), and a clockwise diagonal flip. Dashed lines indicate edges connecting nodes not on $f$.}
    \label{fig:operations}
\end{figure}
Before we describe how to modify $\layout$ to eliminate a critical edge, we first introduce some shorthand notation. Recall that a transversal edge-partition requires the property that, in clockwise order around every non-special node of $\layout$, we have a specific sequence of incoming/outgoing horizontal/vertical edges. We use OV, OH, IV, and IH, to refer to a maximal sequence of outgoing vertical, outgoing horizontal, incoming vertical, and incoming horizontal edges, respectively. An edge sequence of a non-special node is \emph{valid} if it has the form OV $\rightarrow$ OH $\rightarrow$ IV $\rightarrow$ IH (or some cyclic rotation thereof) and is \emph{invalid} otherwise.

\mypar{Collapsing faces.} It is non-trivial to recover a transversal edge-partition after removing or relabeling $e^*$. We therefore use a structural approach that maintains a valid layout guide and must eventually remove or relabel $e^*$. Consider the (blue) faces induced by the blue edges $E_H$ of $\layout$ (see \cref{fig:blueFace}). The critical edge $e^*$ must lie in one of those blue faces. Let $s$ and $t$ be the source and sink nodes of this blue face, respectively. The nodes $s$ and $t$ partition the boundary of the blue face into a lower and upper directed path from $s$ to $t$. The idea is now to start ``merging'' the lower and upper paths, starting from either $s$ or $t$, in a way that resembles the merge step of a merge sort. Intuitively, you can think of this operation as a zipping operation, in which we close the zipper (collapse the face) from $s$ to $t$ or vice versa. The critical edge $e^*$ will be relabeled or eliminated before the face is fully collapsed. 

\begin{figure}[t]
    \centering
    \includegraphics[page=1]{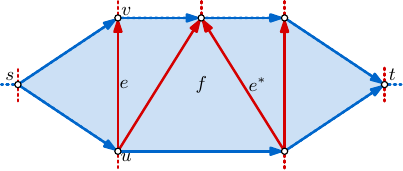}
    \caption{
    The blue face $f$ (shaded) with source node $s$ and target node $t$. Critical edge $e^*$ lies in $f$ and edge $e = (u, v)$ is the leftmost edge in $f$.}
    \label{fig:blueFace}
\end{figure}

Thus, to describe how to eliminate $e^*$, it suffices to describe how to merge the first edges of the blue face. Without loss of generality, we assume that we collapse the face from $s$ to $t$, as the other direction works symmetrically. Let $(s, u)$ and $(s, v)$ be the first lower and upper edges of the blue face, respectively, where the edge $e = (u, v)$ must be red. We show how to relabel $e$ to a blue edge. Furthermore, we use the extremal order $\leq_H$ to determine the direction of the edge between $u$ and $v$, ensuring that $u$ and $v$ are ordered horizontally as specified by $\leq_H$. We assume without loss of generality that $u \leq_H v$; the other case is symmetric by vertically flipping the layout guide.

We simply start by relabeling $e$ to a blue edge. This operation may invalidate the edge sequences around $u$ and $v$. If $e$ was not the only IV edge incident to $v$, then the edge sequence of $v$ remains valid; otherwise, we obtain the invalid edge sequence OV $\rightarrow$ OH $\rightarrow$ IH around $v$. For the edge sequence around $u$ there are two cases: (1) if $e$ was the only OV edge incident to $u$, then we obtain the invalid sequence OH $\rightarrow$ IV $\rightarrow$ IH; (2) otherwise, we obtain the invalid sequence OV $\rightarrow$ OH $\rightarrow$ IV $\rightarrow$ IH $\rightarrow$ OH. In both cases we can use the same approach to correct the edge sequence of $u$, and this approach will also correct the edge sequence of $v$ if it was invalid. However, we have to consider two cases based on the number of vertices on the lower path of the blue face. 

\begin{figure*}[t]
    \centering
    \includegraphics[page=5]{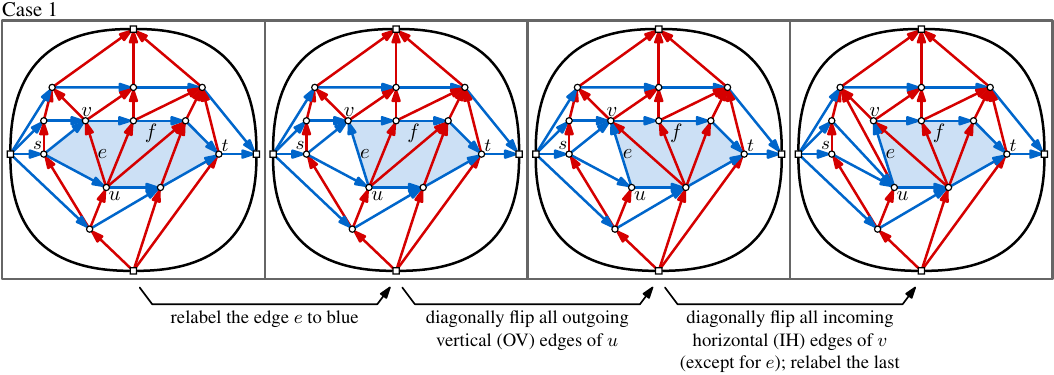} \\
    \vfill
    \includegraphics[page=3]{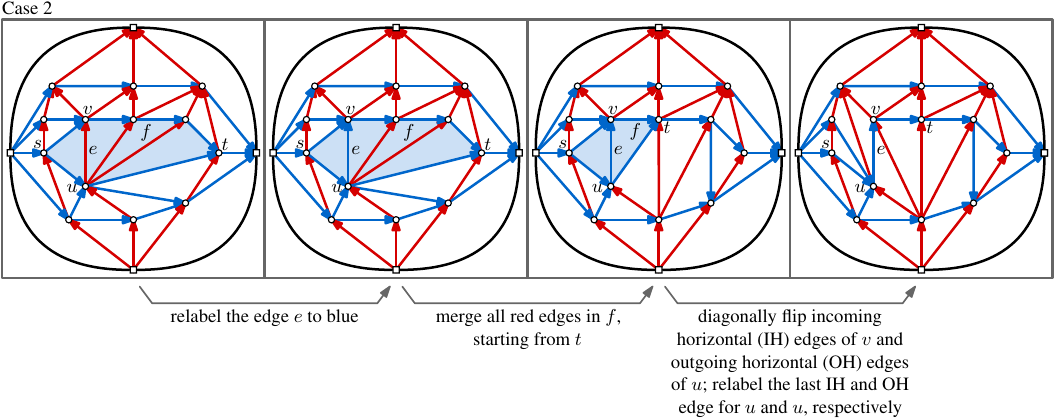} 
    \caption{
    Illustration of the operation sequence to merge the leftmost red edge $e$ in a blue face $f$. Top: Case 1, where $u$ has a neighbor $w \neq s, t$ on the lower blue path. Bottom: Case 2, where $u$ is the only node on the lower blue. }
    \label{fig:merge_cases_u_first}
\end{figure*}

\mypar{Case 1 ($u$ has a neighbor $w \neq s,t$ on the lower path)}

\noindent We first diagonally flip (in CCW direction) all OV edges incident to $u$ in CCW order. This operation ensures that the edge sequence around $u$ becomes OH $\rightarrow$ IV $\rightarrow$ IH, and that the edge sequence around $v$ becomes valid (see~\cref{fig:merge_cases_u_first} top). Furthermore, note that we insert OV edges around $w$, but because $w$ must have already had an OV edge, the edge sequence around $w$ remains valid.

Afterwards, we still need to introduce an OV edge around $u$. This is achieved by diagonally flipping the IH edges around $v$ that are CCW from $e$ (excluding $e$) in CCW order. These edges are diagonally flipped in CW direction, except for the last edge, which is diagonally flipped in CCW direction and relabeled to red. Note that there must be at least one IH edge (other than $e$) incident to $v$, since the edge $(s, v)$ is one of those. Furthermore, it is easy to see that none of the edge sequences around other involved nodes are invalidated. 

\mypar{Case 2 ($u$ is the only node on the lower path)} 

\noindent In this case, we may assume that $v$ is the only node on the upper path. If not, then let $v'$ be the last node on the upper path such that $(v', t)$ is in $\layout$. Since $u \leq_H v$, we must have $u \leq_H v'$. Then we can simply merge the paths starting from $t$ using Case 1 (imagine rotating $\layout$ by $180$ degrees) until $v$ is the only node on the upper path (see~\cref{fig:merge_cases_u_first}~bottom).

In this case, the edge sequence around $u$ must be OH $\rightarrow$ IV $\rightarrow$ IH, and the edge sequence around $v$ must be OV $\rightarrow$ OH $\rightarrow$ IH. We now flip the IH edges around $v$ that are CCW from $e$ (excluding $e$) in CCW order. These edges are diagonally flipped in CW direction, except for the last edge, which is diagonally flipped in CCW direction and relabeled to red. This operation inserts an OV edge around $u$, validating the edge sequence around $u$. Similarly, we flip the OH edges around $u$ that are CW from $e$ (excluding $e$) in CW order. These edges are diagonally flipped in CW direction, except for the last edge, which is diagonally flipped in CCW direction and relabeled to red. This operation inserts an IV edge around $v$, validating the edge sequence around $v$. Furthermore, it is easy to see that none of the edge sequences around other involved nodes are invalidated. 

\smallskip

We continue merging the blue face until the critical edge $e^*$ is relabeled or removed and we have re-established the properties of a transversal edge-partition. Below we prove that the above process must eventually terminate in a layout guide corresponding to a linear horizontal order.

\begin{theorem}
The critical edge elimination procedure described above will reach a layout guide corresponding to a linear horizontal order in a finite number of steps.
\end{theorem}
\begin{proof}
To show that the procedure above always makes progress towards a linear horizontal order, we show that the number of pairs of nodes $x, y$ for which there is a directed horizontal path from $x$ to $y$ in $\layout$ (or vice versa) always increases. Since this number of pairs is maximized in a linear horizontal order, the procedure must eventually reach a layout guide corresponding to a linear horizontal order.

Let $(s, u)$ and $(s, v)$ be the first edges of the blue face that contains the critical red edge $e^*$. We are merely interested in how the horizontal (blue) edges change. In the procedure described above, $v$ may lose some IH edges $(w, v)$, but only after the IH edge $(w, u)$ was already introduced. Since we add the horizontal edge $(u, v)$, a directed horizontal path between $w$ and $v$ is preserved. Similarly, $u$ may lose some OH edges $(u, w)$, but only after the OH edge $(v, w)$ was already introduced. Hence, a directed horizontal path between $u$ and $w$ is again preserved. The pair $u, v$ was previously not connected by a directed horizontal path, but this connection is introduced via the edge $(u, v)$. Thus, the number of pairs of nodes connected by a directed horizontal path in $\layout$ must strictly increase, which completes the proof.
\end{proof}

\noindent Note that we do not claim that the final linear horizontal order corresponds to $\leq_H$. Because we require the layout guides to be stable, we cannot guarantee that the extremal order $\leq_H$ is obtained as the final linear horizontal order of map elements, if the reference layout $\refLayout$ and extremal order $\leq_H$ are not very compatible. However, for suitable choices of $\refLayout$ and $\leq_H$, the layout guides usually converge to a linear order close to $\leq_H$. In fact, in our use cases we always obtain the extremal order exactly in the limit (see \cref{sec:instantiations}).

\subsection{Merge Heuristics}\label{sec:critical-edge-heuristics}
In this section we describe three possible heuristics for choosing the critical edge $e^*$ from a critical (vertical) path $p^* \in P^\uparrow$, and discuss advantages and drawbacks of each. We assume that there is only one critical path $p^*$. Note that, in degenerate cases, there may be multiple critical paths. In that case, it may be beneficial to choose critical edges that are contained in multiple critical paths. However, we do not consider this situation here. 

When choosing a critical edge, we are concerned with three aspects: (1) the amount of change (either in total area or number of involved elements) required to eliminate the critical edge, (2) the reduction in height achieved by eliminating the critical edge, and (3) the balance in the distribution of map elements. Our heuristics are designed to incorporate these different aspects. Observe that we can eliminate only edges between non-boundary nodes, so edges in $p^*$ adjacent to boundary nodes are not considered. This also implies that we should only consider critical paths containing at least two non-boundary nodes; we therefore ignore all other paths in the computation of $w(\layout)$ and $h(\layout)$, ensuring that $p^*$ contains at least one critical edge between two non-boundary nodes.
Finally, note that we can also choose in which direction we can collapse a face (from $s$ or from $t$). This choice is only relevant in our first heuristic; therefore, our other heuristics also optimize this choice based on the first heuristic. 

\mypar{Min-change.} This heuristic selects the critical edge $e^* \in p^*$ that can be eliminated in the fewest number of merges. To establish that number for a specific edge $e^*$, we consider the blue face in which $e^*$ is contained, and count the number of red edges between $e^*$ and either $s$ or $t$ (taking the minimum). In case of ties, we can break ties arbitrarily. However, recall that we must always break ties consistently/deterministically, for otherwise the resulting layout guides are not stable and consistent.

The advantage of this heuristic is that it leads to the smallest number of combinatorial changes to $\layout$. However, it does not consider the size of the map elements, and hence the visual impact of the changes. For that reason we can also consider the total area of the affected map elements instead, and minimize that. If we also want to minimize the number of combinatorial changes, then we can choose a combined measure. One option is to simply minimize the number of merges, and break ties by considering the total area of the affected map elements.

\mypar{Max-height.} This heuristic selects $e^*$ as one of the two edges adjacent to the node $v$ with maximum $h(v)$. Unless one of these edges is incident to a special node, this gives us two possible choices $(u, v)$ and $(v, w)$ for $e^*$. In that case, we choose $e^* = (u, v)$ if $h(u) > h(w)$, and $e^* = (v, w)$ otherwise.

The advantage of this heuristic is that it is likely to have the biggest impact on the height of the resulting layout guide, potentially resulting in fewer critical edges that need to be eliminated. However, this heuristic will often make changes to large map elements, which can have a high visual impact on the layout.

\mypar{Min-width.} This heuristic aims at obtaining a balanced distribution of the map elements over the resulting map. Note that our algorithm focuses on reducing the height of the layout guide, but does not control the width of the layout guide. By eliminating a critical edge $e^* = (u, v)$, we essentially place the map elements corresponding to $u$ and $v$ horizontally next to each other. Therefore, several horizontal paths in $P^\rightarrow$ become longer. This heuristic selects the critical edge $e^*$ that minimizes the width $w(p)$ of the widest horizontal path $p \in P^\rightarrow$ that contains $u$ and $v$ in the resulting layout $\layout$, where $w(p) = \sum_{v \in p} w(v)$.

The advantage of this heuristic is that it tries to minimize the width of the layout, which effectively increases the area that can be used by the map elements. A potential drawback is that it may perform more changes than necessary to the layout.

\section{Use Cases}
\label{sec:instantiations}

Layout guides capture the essential aspects of thematic maps. We hence expect that most, if not all, thematic mapping algorithms can be adapted to follow such a guide and produce a thematic map that closely matches the relative positions and visual requirements of the statistical map elements encoded therein. This adaptation is particularly straightforward for cartograms; we demonstrate the necessary steps for both rectangular and Demers cartograms. We do so using three datasets, which are vertical (the regions of England\footnote{https://en.wikipedia.org/wiki/Regions\_of\_England}), horizontal (the contiguous United States\footnote{https://en.wikipedia.org/wiki/List\_of\_U.S.\_states\_and\_territories\_by\_area}), and square (the departments of France\footnote{https://en.wikipedia.org/wiki/List\_of\_French\_departments\_by\_population}). In each case, the thematic value assigned to a region corresponds to its (land) area.

We describe how to initialize the map arranger for the three datasets, followed by describing our thematic mapping algorithms for generating rectangular cartograms and Demers cartograms from layout guides. Our code is openly available.\footnote{https://github.com/cartocrow/responsive\_map\_arranger/tree/ArXiv} 

To offer some more flexibility in using the map arranger, we first introduce the notion of \emph{slack}. By design, the map arranger is very rigid when determining whether a layout guide fits in a container. However, in practice, depending on the type of thematic map, we may want to allow a slight deformation in the shape of map elements in favor of maintaining their relative positions. We therefore introduce a slack parameter $s \geq 0$ that controls this tradeoff. Specifically, if $\container$ is the given container and $s$ is the slack parameter, then the map arranger obtains $\layout$ by (starting from $\refLayout$) eliminating vertical (or horizontal) critical edges until $h(\layout) < (1+s) h(\container)$ (or $w(\layout) < (1+s) w(\container)$). The setting $s = 0$ corresponds to the original rules of the map arranger.

%For this, we first introduce the notion of \emph{slack}. Without slack, the map arranger always merges critical edges of a reference layout $\refLayout$ until there is no critical path left. However, often, in responsive visualization, it is desirable to squish map elements a bit before reordering them. We can therefore only resolve critical paths of $\refLayout$ once they exceed the container dimensions by a certain threshold. The slack is a percentage of the container $\container$ dimensions; for a given slack $s$, vertical critical paths are merged until the resulting layout guide $\layout$ its height is $h(\layout) <h(\container) + s\cdot h(\container)$ and horizontal critical paths are merged until $w(\layout) < w(\container) + c\cdot w(\container)$.

\subsection{Initializing the Map Arranger}
To initialize the map arranger, a reference layout $\refLayout$ must be provided together with a horizontal and vertical extremal order. Our general approach is as follows. Since we will construct cartograms for geographic regions as map elements, we will use the geographic regions in a general reference map as the basis to determine $\refLayout$, where each node in $\refLayout$ corresponds to a geographic region. Each map element is assigned a desired aspect ratio and relative area, which is derived from the bounding box of the corresponding geographic region (the relative area is with respect to the total area of all map elements). To better preserve the original outline of the maps in the resulting cartograms, we sometimes introduce additional \emph{sea regions} \cite{kreveld2007rectangular} as map elements. These sea regions help maintain the overall shape of the map and support the construction of an appropriate triangulation for $\refLayout$. As explained in Section~\ref{sec:map-arranger}, we also need to add the four required boundary nodes W (West), E (East), S (South), and N (North) to $\refLayout$, which correspond to the left, right, bottom, and top side of the map, respectively. We now discuss the remaining details for constructing $\refLayout$ for each of the datasets separately.

\mypar{Regions of England.} To construct $\refLayout$ for this map, we first add some sea regions to the map. Specifically, to preserve the characteristic L-shapes on the left and top-right parts of the map, one sea region is added in each of these areas. Additionally, two sea regions are introduced around London to facilitate the construction of the irreducible triangulation (see \cref{fig:uk_initialization} middle). Next, we can obtain an irreducible triangulation by taking the dual graph of the geographic regions (two nodes are connected if their corresponding regions share a boundary). Finally, we construct a transversal edge-partition on the triangulation, which can then be used as $\refLayout$ (see \cref{fig:uk_initialization}, right).
%To initialize $\refLayout$, we first add a vertex that corresponds to each region (see~\cref{fig:uk_initialization} left). We then add the four outer vertices W, E, S, and N. To preserve the characteristic L-shapes on the left and top-right parts of the map, one sea region is added in each of these areas. Additionally, two sea regions are introduced around London to facilitate the construction of the irreducible triangulation shown in the middle of \cref{fig:uk_initialization}. 
%Once this triangulation is established, the transversal edge partition is constructed; see \cref{fig:uk_initialization}, right.
\begin{figure*}[t]
    \centering
    \includegraphics{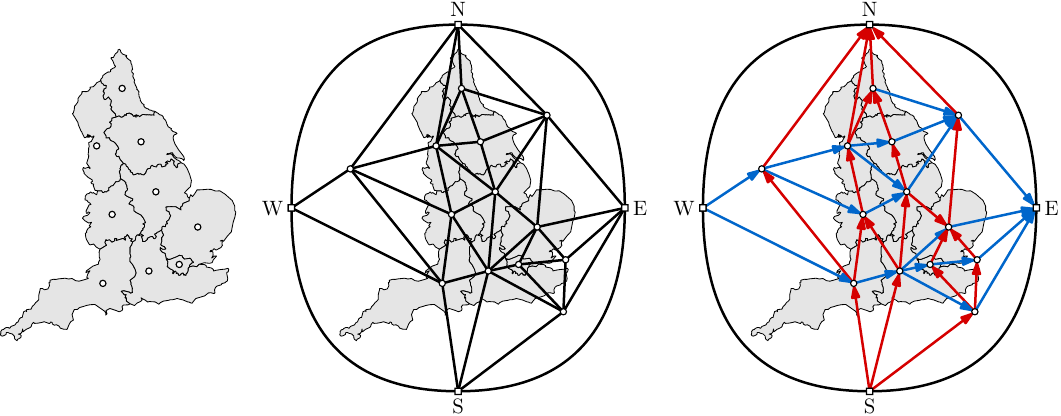}
    \caption{Constructing a reference layout $\refLayout$ for the regions of England. Left: map elements are represented as nodes storing relative area and aspect ratio. Middle: the four boundary nodes (W, E, S, and N) and four additional sea regions are added; the resulting dual graph is an irreducible triangulation. Right: a transversal edge-partition is computed on the triangulated graph.}
    \label{fig:uk_initialization}
\end{figure*}

\mypar{Departments of France.} For the departments of France, we add six sea regions that help preserve the outer shape of France and ensure that the resulting reference layout captures the relative positions of regions on the map. Unfortunately, the dual graph of the resulting regions contains a separating triangle formed by the three regions surrounding Paris. To obtain an irreducible triangulation, we simply merge the Paris region with its three surrounding regions. Finally, we construct a transversal edge-partition on the dual graph (triangulation), which can then be used as $\refLayout$ (see the supplementary material). 
%the initialization of $\refLayout$ slightly differs. Paris, together with the three surrounding regions, creates a separating triangle. Instead of assigning each of these regions a separate node, they are merged into a single region. The area of this merged region is set to the sum of the individual areas. All remaining regions are assigned their own node. Then, the four outer nodes are added, along with six sea regions that help preserve the outer shape of France and ensure that the resulting structure reflects the spatial relations of regions in the map. Afterward, the graph is triangulated based on adjacency, and the transversal edge partition that can be seen in the supplementary material is created.

\mypar{Contiguous United States.} We add eleven sea regions to help preserve the outer shape of the US, and to obtain an irreducible triangulation. The resulting dual graph does not contain separating triangles, but the US states include a point in which four states meet. To obtain a valid triangulation, we must drop one of these adjacencies. Finally, we construct a transversal edge-partition on the resulting triangulation, which can then be used as $\refLayout$ (see the supplementary material).
%The initialization of the US states again includes a node for each region. The four outer nodes are added, in addition to eleven sea regions. As New York is a cut vertex in the adjacency graph of the US states, at least the sea regions surrounding New York are required to be included in order to create an irreducible triangulation. After adding all nodes the graph is triangulated. Note that the US states include a point in which four states border. As a triangulation cannot capture these four adjacencies, one adjacency is not included in the resulting triangulation. Afterward, the transversal edge partition is initialized. The triangulation and transversal edge partition can be seen in the supplementary material.

\mypar{Extremal orders.} To choose a suitable extremal order $\leq_H$ for a specific map, we first consider the horizontal edges in $\refLayout$. We ensure that if there is a horizontal directed path from node $u$ to node $v$, then $u \leq_H v$. Under these constraints, we construct a Hamiltonian path on the triangulation induced by $\refLayout$, excluding the sea regions when possible. The extremal order $\leq_H$ of the land regions then simply corresponds to the order along the Hamiltonian path in which the land regions are encountered (for example, see \cref{fig:uk_rect_output}). We order the sea regions at either end of the order $\leq_H$, with the aim of keeping the land regions contiguous in the resulting rectangular cartogram. The vertical extremal order $\leq_V$ is constructed symmetrically. See the supplementary material for the extremal orders of France and the US.

%For the extremal orders we provide, for all three datasets, a horizontal and vertical extremal order that is compatible with the partial order of the reference layout $\refLayout$ and is a Hamiltonian path in its triangulation. The total order provided for the land regions of England can be seen on the left of \cref{fig:uk_rect_output}. We insert sea regions on the outside of the orders in order to make sure the land regions remain contiguous in the resulting rectangular cartogram. The total orders provided for France and the US are provided as part of the supplementary material.

%Below, we explain how we initialize the reference layout $\refLayout$ for each of the datasets separately.

% \begin{figure}
%     \centering
%     \includegraphics{figs/UK_order_construction.pdf}
%     \caption{ The triangulated graph for the reference layout $\refLayout$ of the regions of England, together with the provided extremal orders. Blue arrows indicate the extremal horizontal order $\leq_H$ and red arrow indicate the extremal vertical order $\leq_V$. The outer vertices are excluded from the triangulation as they are not relevant for the extremal orders.}
%     \label{fig:uk_orders}
% \end{figure}

\subsection{Layout Guide to Rectangular Cartogram}
A \emph{rectangular layout} is a partitioning of a rectangle into smaller rectangles. A rectangular cartogram \cite{kreveld2007rectangular, raisz1934rectangular} is a rectangular layout where each map region is represented by a rectangle, such that the area of the rectangle corresponds to the thematic data value of the map region. 

It is known that a transversal edge-partition of an irreducible triangulation directly corresponds to a rectangular layout where two rectangles share a vertical boundary if there is a horizontal edge between the corresponding nodes, and a horizontal boundary if there is a vertical edge between the corresponding nodes~\cite{fusy2009transversal, kant1997regular}. However, the resulting rectangular layout may not function as a correct rectangular cartogram, as the areas may be incorrect. Nevertheless, for any rectangular layout and assignment of areas to rectangles, there always exists an \emph{order-equivalent} rectangular layout that realizes the given areas~\cite{eppstein2012area}. Specifically, this order-equivalent rectangular layout preserves the horizontal and vertical order among rectangles, but may change direct rectangle adjacencies.  

Thus, to construct a rectangular cartogram with a given layout guide $\layout$, we first construct a rectangular layout from the transversal edge-partition of $\layout$ using the algorithm in \cite{kant1997regular}. Next, we use the hill climbing algorithm in \cite{eppstein2012area} to construct an order-equivalent rectangular layout with the correct areas, which will be the resulting rectangular cartogram.  
%It is known that the dual of a transversal edge partition $\mathcal{T}$ equates to a partition of rectangles such that two rectangles partially share a boundary if and only if their corresponding nodes in $\mathcal{T}$ share an edge \cite{fusy2009transversal, kant1997regular}, but not all such partitions admit a realization of an arbitrary assignment of positive areas. There is however, always a class of rectangular layouts that are \emph{order equivalent} \cite{eppstein2012area}, of which there is one layout that does realize an arbitrary assignment of areas. Two rectangular layouts are order equivalent if they induce the same horizontal and vertical partial orders on the \emph{maximal segments}; a segment is maximal if it is not contained in any other segment of the rectangular layout.
%Using these properties, we can simply construct a rectangular layout from the transversal structure of a layout guide $\layout$, given by the map arranger, using the algorithm in \cite{kant1997regular} and then realize the area assignments via the gradient descent algorithm from \cite{eppstein2012area} to construct a rectangular cartogram.

\mypar{Showcase.} We present the results for rectangular cartograms of the regions of England. Results for France and the US are included in the supplementary material. \cref{fig:uk_rect_output} shows rectangular cartograms generated using layout guides that are produced by the map arranger with the min-width heuristic. It displays snapshots from videos in the supplementary material for seven containers of equal area for various aspect ratios. 

All regions remain clearly visible across all ratios; no region becomes too small or thin to read. For both extreme aspect ratios, the ordering of map elements in the resulting cartograms aligns with the predescribed extremal orders.

Comparing the min-width heuristic with the min-change heuristic shown in the supplementary material, we observe that the min-width heuristic distributes regions more evenly when the container is compressed by 50\%. In contrast, the min-change heuristic already produces a linear order of land regions under horizontal compression. The formation of long paths is even more prevalent in the larger examples shown in the supplementary material.

Another aspect concerns the sea regions. These currently occupy a substantial amount of space under extreme aspect ratios without adding any information or supporting the outer shape. How to address this depends on the use case and should be decided by the thematic mapping algorithm; we discuss this further in \Cref{sec:discussion}.

\begin{figure*}[t]
    \centering
    \includegraphics[page=1]{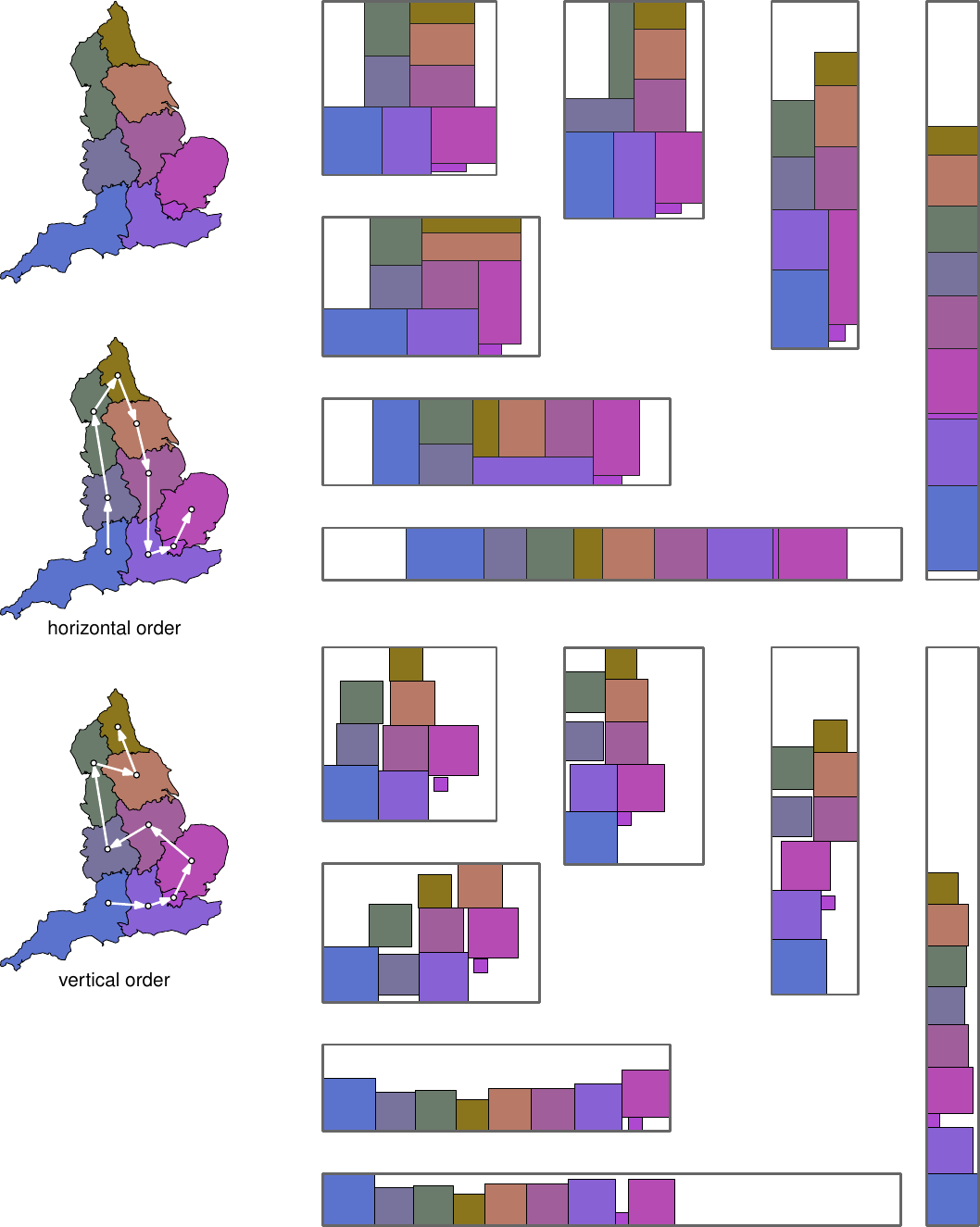}
    \caption{Rectangular and Demers cartograms produced using layout guides from the map arranger using the min-width heuristic (right) of the regions of England (top left). We used a slack of 0.3 for the rectangular cartograms (top), and a slack of 0.2 for the Demers cartograms (bottom). Both cartogram outputs are shown for various aspect ratios: the top-left container has an aspect ratio of 1:1. The other containers are compressed horizontally and vertically by 20\%, 50\%, and 70\% and then stretched in the other direction such that the area of all containers is equal. The extremal horizontal and vertical orders of the land regions are indicated on the left.}
    \label{fig:uk_rect_output}
\end{figure*}

\subsection{Layout Guide to Demers Cartograms}

Demers cartograms represent each region of interest by a square, with its area matching an associated data value of interest; these squares are to be positioned without overlap while maintaining geographic aspects. 
Nickel et al.~\cite{nickelMulticriteriaOptimizationDynamic2022} introduced a linear program to compute Demers cartograms, based on a set of separation constraints (horizontal or vertical), such that any two regions are transitively separated either vertically or horizontally. After heuristically estimating these separation constraints, this linear program then minimizes the Manhattan distance between adjacent regions.
To derive a Demers cartogram from a layout guide, we simplify, adapt and extend this linear program; refer to \iflabelexists{app:demers}{the supplementary material~\ref{app:demers}}{the supplementary material} for details.

First, we simplify the linear program to avoid various secondary objectives that are introduced by Nickel et al.~\cite{nickelMulticriteriaOptimizationDynamic2022}: while we could integrate such considerations also here, we choose to explicitly focus on the layout guide only as a means of determining the eventual cartogram.

Second, we adapt the algorithm to use the labeling of the layout guide to prescribe the separation constraints, instead of heuristically inferring them. However, these constraints are not quite sufficient to ensure that the Demers cartogram is free from overlap. After solving the linear program, we check for overlapping elements: if there are any, we add further separation constraints and repeat.
In our test cases, only a few iterations, if any, were observed.

Third, we extend the linear program to maximize symbol size, before minimizing distances between regions. As the original method uses fixed symbol sizes without a specific container, this is necessary to ensure that the resulting cartogram uses the available space effectively without having to adapt the thematic data. To achieve maximizing symbol sizes, we integrate the scale maximization in the linear program for overlap removal by Meulemans~\cite{meulemansEfficientOptimalOverlap2019} into our algorithm.

\mypar{Showcase.} For Demers cartograms, the results of the min-width heuristic applied to the regions of England are also presented in \cref{fig:uk_rect_output}; results for France and the US are included in the supplementary material. As with the rectangular cartograms, each region remains clearly visible, and the total orders are consistent with the provided extremal orders. Note that the total orders generated by the map arranger do not depend on the thematic mapping algorithm used.

The relative positions of the regions and outer shape of the map is maintained quite well when the map is compressed vertically by 20\% (i.e., to a landscape aspect ratio). However, under more extreme compression (50\%), the vertical space is not utilized optimally: some regions could still be stacked on top of each other. These regions were part of a critical path at less extreme aspect ratios, causing them to merge into a horizontal order earlier.

Comparing the min-width heuristic with the min-change heuristic, shown in the supplementary material, we observe the same uneven distribution issues for the min-change heuristic that were also present in the rectangular cartograms. The distribution for the min-width heuristic remains even across the container. Note again that the partial orders defined by the layout guide do not depend on the thematic mapping algorithm, but solely on the map arranger.

\section{Discussion and Conclusion}
\label{sec:discussion}

We introduced an algorithmic framework to support the responsive construction of thematic maps. Its core element is a map arranger, that adapts a layout guide (relative directions between map elements) to a container based on the desired width and height of the map elements. The map arranger ensures that the layout guides adapt in a stable and consistent manner, ensuring small changes in the guide for small changes to the container and that identical containers result in identical layout guides. To use this framework, one needs to introduce or adapt thematic mapping algorithms to use the layout guide: we demonstrated how to do so for rectangular cartograms and Demers cartograms. Our work allows for various avenues for further research, briefly discussed below.

\mypar{Sea regions \& outer shape.}
We introduced sea regions in our datasets to help shape the resulting thematic maps, particularly the overall outline of the composition of map elements. These regions may thus help with recognizing the geographic regions and location-specific elements. Yet, as the aspect ratio changes, or the container becomes smaller, affording the extra space to such sea regions may not weigh up against the reduced space for the actual map elements. We leave effectively adapting sea regions within the framework to future work, but observe that doing so may depend on geographic context as well as on the thematic map itself. For example, maintaining a sea region to represent the Atlantic Ocean may be important in a world map; but a sea region to maintain more shape to a UK map may reduce in importance as space is constrained more.

\mypar{Other thematic maps.}
Our framework can be used to make responsive versions of many different types of thematic maps, by developing appropriate thematic mapping algorithms. 
Two promising examples are mosaic cartograms and grid maps, which have been increasing in popularity in recent years, particularly for election mapping. Automated generation of both mosaic cartograms~\cite{canoMosaicDrawingsCartograms2015} and grid maps~\cite{meulemansSimplePipelineCoherent2021} yields good results, but existing algorithms would need to be adapted to be compatible with the layout guide.

\mypar{Efficiency.}
Our map arranger is efficient and can easily support real-time adaptation of a layout guide. The main bottleneck in efficiency hence lies with the thematic mapping algorithm itself. In our use cases, these algorithms were also sufficiently fast to enable on-the-fly computation of the map as the container and layout guide change. For other map types, such as mosaic cartograms mentioned above, the algorithm may not be fast enough for real-time updates. Instead, one could precompute the result for several layout guides, observing that there are only a limited number of guides that the map arranger can produce, possibly complemented by a heuristic to fit the precomputed layout to an exact container.

\mypar{Animated transitions.}
Animated transitions help users keep track of changes in visualizations~\cite{chevalierNotsoStaggeringEffectStaggered2014}. As such, animating changes as a thematic map is gradually rearranged for different aspect ratios could be a useful addition for responsive visualization scenarios where users resize the map frequently. Our framework makes discrete changes to the layout guide, but was designed for those changes to be gradual. It is therefore well-suited for the development of thematic maps with animated transitions of moving map elements.

% \subsection{User agency}

% Layout guide identifies when things don't fit and why they need to move --- could be used to inform users about these changes

\mypar{Consistency \& user studies.}
While \textit{consistency} is a common design goal in responsive visualization (see e.g.,~\cite{hoffswellTechniquesFlexibleResponsive2020,wuViSizerVisualizationResizing2013}), it is unclear to what extent existing design guidance on consistent encodings~\cite{quKeepingMultipleViews2018} applies to responsive visualizations. We aim to ensure consistency by maintaining relative positions and gradually adjusting the spatial layout of map elements. Future work should aim to empirically evaluate whether users perceive these changes to be consistent and if it helps them keep track of map elements and locate regions even in a distorted map.

\clearpage

\acknowledgments{
A. Simons and W. Meulemans are (partially) supported by the Dutch Research Council (NWO) under project number VI.Vidi.223.137.}

\bibliographystyle{abbrv-doi-hyperref-narrow}
% %\bibliographystyle{abbrv-doi}
% %\bibliographystyle{abbrv-doi-narrow}

\bibliography{responsive-maps-framework}

\clearpage
%\appendix

\pagestyle{empty}

\begin{figure*}[t]
\begin{minipage}{\textwidth}
\centering
{\sffamily\huge%
Automated Responsive Thematic Mapping with Layout Guides
\vskip 5pt
\LARGE Supplementary Material\\}
\vskip 25pt
{\sffamily\large%
Arjen Simons, Sarah Schöttler, Wouter Meulemans, Kevin Verbeek, and Bettina Speckmann\\
\vspace{25pt}
}
\end{minipage}
\end{figure*}

\setcounter{section}{0}

\section{Linear program for Demers cartograms}\label{app:demers}

Specifically, consider a Layout Guide $\mathcal{L} = (V,E)$ where the desired width and height for every vertex are set to be equal.
Assume our target rectangle have corners $(0,0)$ and $(a,b)$ for some positive $a,b$. We introduce a single variable $S$ to capture scale, variables $X_v, Y_v$ for each vertex $v$ to capture its position, and variables $\Delta X_{u,v}, \Delta Y_{u,v}$ for each edge $(u,v) \in E$ to capture the distance between vertices connected by an edge in the Layout Guide.

Our linear program has the constraints below. The first four ensure that placing a square of with the desired width (and height) at the center of the vertex coordinates keeps the square within the target rectangle. 
The next two constraints ensure that neighboring vertices in $\mathcal{L}$ are positioned such that their squares do not properly intersect, enforcing either a horizontal or vertical separation depending on the edge labeling. 
The final four constraints ensure that, when minimized, the $\Delta$ values indeed capture the desired distance.
\begin{align}
    X_v - \frac{w(v)}{2} \cdot S \geq 0 & ~~~\forall v \in V \\
    X_v + \frac{w(v)}{2} \cdot S \leq a & ~~~\forall v \in V \\
    Y_v - \frac{w(v)}{2} \cdot S \geq 0 & ~~~\forall v \in V \\
    Y_v + \frac{w(v)}{2} \cdot S \leq b & ~~~\forall v \in V \\    
    X_v - X_u \geq \frac{w(u)+ w(v)}{2} \cdot S & ~~~\forall (u,v) \in E_H \\
    Y_v - Y_u \geq \frac{w(u)+ w(v)}{2} \cdot S & ~~~\forall (u,v) \in E_V \\ 
    \Delta X_{u,v} \geq X_u - X_v -  \frac{w(u)+ w(v)}{2} \cdot S & ~~~\forall (u,v) \in E\\
    \Delta X_{u,v} \geq X_v - X_u  -  \frac{w(u)+ w(v)}{2} \cdot S & ~~~\forall (u,v) \in E\\
    \Delta Y_{u,v} \geq Y_u - Y_v -  \frac{w(u)+ w(v)}{2} \cdot S & ~~~\forall (u,v) \in E\\
    \Delta Y_{u,v} \geq Y_v - Y_u -  \frac{w(u)+ w(v)}{2} \cdot S & ~~~\forall (u,v) \in E
\end{align}

The objective of our linear program is to minimize the function below. Here, we prioritize maximizing scale $S$ by setting $f = -|E|(a+b)$, followed by minimizing the total distance between neighboring vertices in $\mathcal{L}$.
\[ f \cdot S + \sum_{(u,v) \in E} \Delta X_{u,v} + \Delta Y_{u,v} \]

We observe that the constraints derived from the Layout Guide are not quite sufficient to ensure that the result is free from overlap. Vertices that are connected via a directed path with the same label will be overlap free, but not all pairs satisfy this property necessarily.
While we could derive additional constraints upfront, we rather employ a technique often applied for Integer Linear Programming: lazy constraints. 

That is, we solve the above program, and detect whether there are any overlaps. If there are none, we are done. Otherwise, for each overlap, we add a separation constraint: if their coordinates differ more in the vertical direction, we add a vertical separation constraint (as for edges in $E_V$) and a horizontal separation constraint otherwise (as for edges in $E_H$). Note that we do not add $\Delta$ variables for these pairs, as we are not specifically aiming to keep them close. Then, we re-solve the linear program and test for overlap again. This process must terminate as eventually all pairs will have a separation constraint (either directly or transitively) and result in an overlap-free Demers cartogram. 

\section{Results}
\label{app:results}

In this section, we showcase and discuss the results of our responsive mapping framework, via its application to rectangular and Demers cartograms.

We fix the slack to 0.3 for rectangular cartograms and to 0.2 for Demers cartograms, as these values work well for our instances overall. The rectangular cartograms use slightly more slack as the eventual shape of each region is more flexible, compared to the rigid use of squares in the Demers cartogram. We specifically avoid further fine-tuning parameters on instance level: though such practice can improve results on an individual basis, it may obfuscate benefits and drawbacks and prevent generalization of the results. 

The online supplementary material includes a video\footnote{https://youtu.be/EZN7brOWAv4}, demonstrating these responsive rectangular and Demers cartograms, for all three datasets and both the min-change and min-width heuristics. Below, we discuss the results in this video and the snapshots thereof in the corresponding figures in this appendix.
Note that all figures follow the same structure, as also used in the main paper: the top-left container has an aspect ratio of 1:1; the other containers are compressed horizontally and vertically by 20\%, 50\%, and 70\% and then stretched in the other direction such that the area of all containers is equal.

\mypar{Regions of England.}
\Cref{fig:uk_rect_output-appendix} and \Cref{fig:uk_minchange_output} shows the results of our framework for the min-width and min-change heuristic for comparison. Note that the former is the same as Figure~\ref{fig:uk_rect_output} in the main paper, replicated here to facilitate comparison.

We observe that the square and the first horizontally compressed result match the overall shape and structure of the geographic map quite well, given the constraints. This is likely due to the geographic map being slightly taller than wide. Furthermore, the limit in both directions achieves the provided extremal orders.

However, for the other container sizes, we see that the min-width heuristic achieves layouts with lower overall aspect ratio (rectangular cartogram) or larger squares (Demers cartogram) compared to the min-change heuristic. 

We do observe that the Demers cartograms do not always use the space effectively. For example, the two smallest regions in the 50\% vertically compressed layout could easily be on above or below their neighboring regions. While this particular instance can possibly be improved by using a different slack parameter, we note that this is also an effect of ensuring that layout guides are stable and consistent: once a choice is made, it is no long reverted. Furthermore, by the inherent design of squares only, fully occupying the container in both directions is generally not the optimal result.

In the video, we may observe a difference between the rectangular and Demers cartograms (in fact, for each of the datasets): where the former changes smoothly while the layout guide stays the same, the latter occasionally exhibits quite abrupt changes. This is caused by the linear program, being based purely on the layout guide, admits many optimal solutions. For example, when the container is not fully used, any translation that keeps the result within the container achieves the same objective value. As a result, minor changes in the formulation can cause the linear-program solver to select a different result. Adding secondary objective functions can alleviate this issue.

\mypar{Departments of France.}
\Cref{fig:france-transversal} shows the constructed reference layout, in which six sea regions were added to support capturing its shape in the resulting visualizations.

\Cref{fig:france_min-width_output} and \Cref{fig:france_min-change_output} shows the results of our framework for the min-width and min-change heuristic for comparison.

We observe that for the square and near-square results match the overall shape and structure of the geographic map quite well, given the constraints. This is likely due to the geographic map overall quite square as well. Furthermore, the limit in both directions achieves the provided extremal orders (not shown in figures).

However, for the other container sizes, we see that the min-width heuristic achieves layouts with lower overall aspect ratio (rectangular cartogram) or larger squares (Demers cartogram) compared to the min-change heuristic. We specifically see the formation of long paths in the layout guide, manifesting as stacks (horizontally compressed) or rows (vertically compressed) of narrow rectangles or squares, in the min-change heuristic. The min-width heuristic does not exhibit this behavior as strongly. However, in the rectangular cartogram

\mypar{Contiguous United States.}
\Cref{fig:USA-transversal} shows the constructed reference layout, in which eleven sea regions were added to support capturing its shape in the resulting visualizations. 

\Cref{fig:usa_min-width_output} and \Cref{fig:usa_min-change_output} shows the results of our framework for the min-width and min-change heuristic for comparison.

We observe that for the square and first vertically compressed results match the overall shape and structure of the geographic map quite well, given the constraints. This is likely due to the geographic map overall being slightly wider than tall. Furthermore, the limit in both directions achieves the provided extremal orders on the actual regions (not shown in figures).

Between the two heuristics, we can largely make the same observations here as for the previous dataset: the emergence of long paths in the layout guide cause for many narrow rectangles or stacks of small squares -- in this dataset even more so that in the previous ones. Hence, we prefer the min-width heuristic over the min-change heuristic overall as it consistently performs better. 

We do observe that in these results, the use of sea regions leads to undesirable behavior, as they appear in places with the goal of providing shape while the container is compressed beyond the levels that could sustain something resembling the overall map. Moreover, it even disconnects the contiguous states. 
As also discussed in the main paper, however, we leave adapting the constraints for the sea regions in a responsive manner is left to future work.

\begin{figure*}[p]
    \centering
    \includegraphics[page=1]{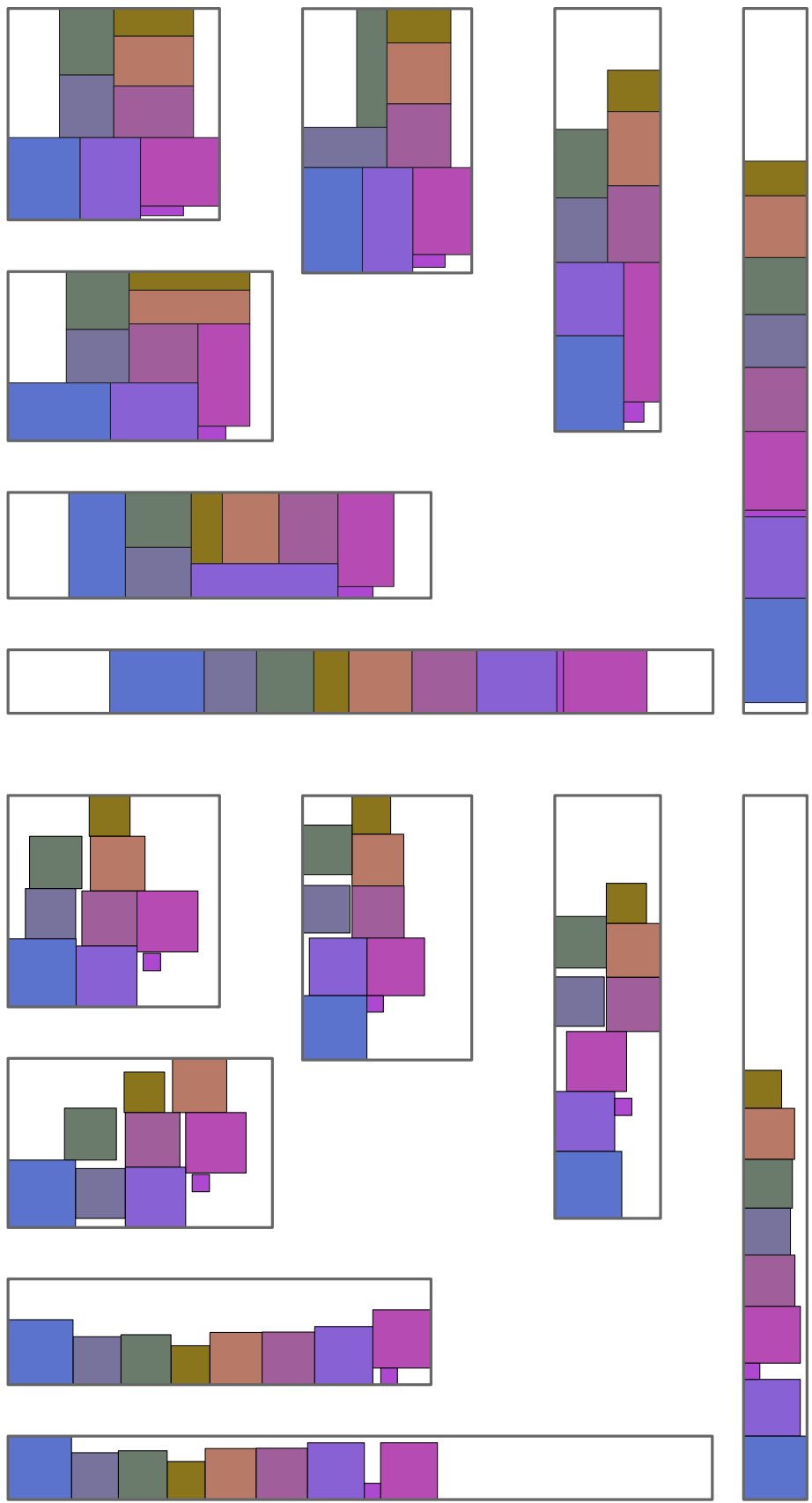}
    \caption{Results for the min-width heuristic on the regions of England. (left) Geographic map and its extremal horizontal and vertical orders of the land regions. (top) Rectangular cartograms produced using layout guides from the map arranger. (bottom) Demers cartograms produced using layout guides from the map arranger.}
    \label{fig:uk_rect_output-appendix}
\end{figure*}

\begin{figure*}[p]
    \centering
    \includegraphics[page=2]{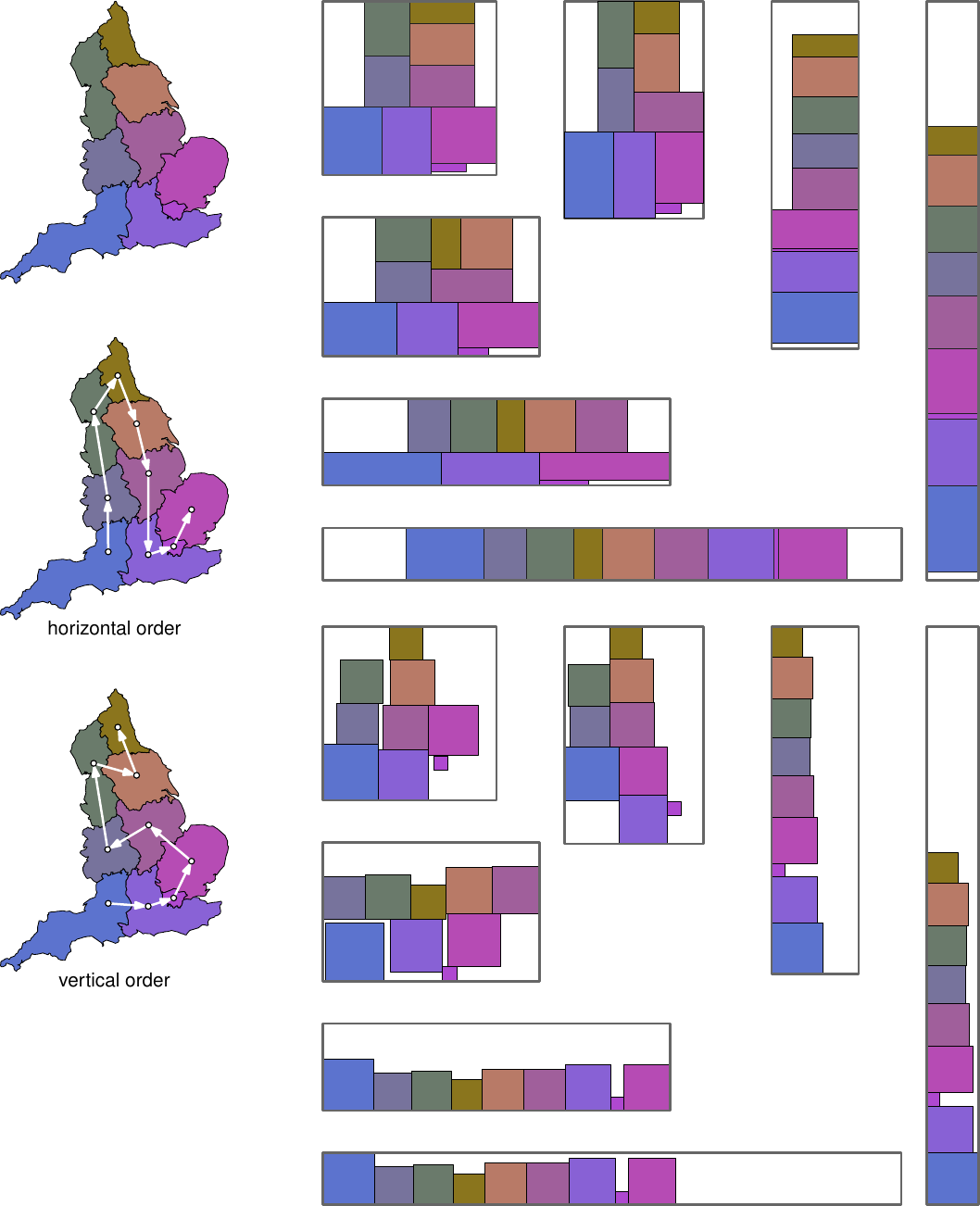}
    \caption{Results for the min-change heuristic on the regions of England. (left) Geographic map and its extremal horizontal and vertical orders of the land regions. (top) Rectangular cartograms produced using layout guides from the map arranger. (bottom) Demers cartograms produced using layout guides from the map arranger.}
    \label{fig:uk_minchange_output}
\end{figure*}

\begin{figure*}
    \centering
    \includegraphics[page=1]{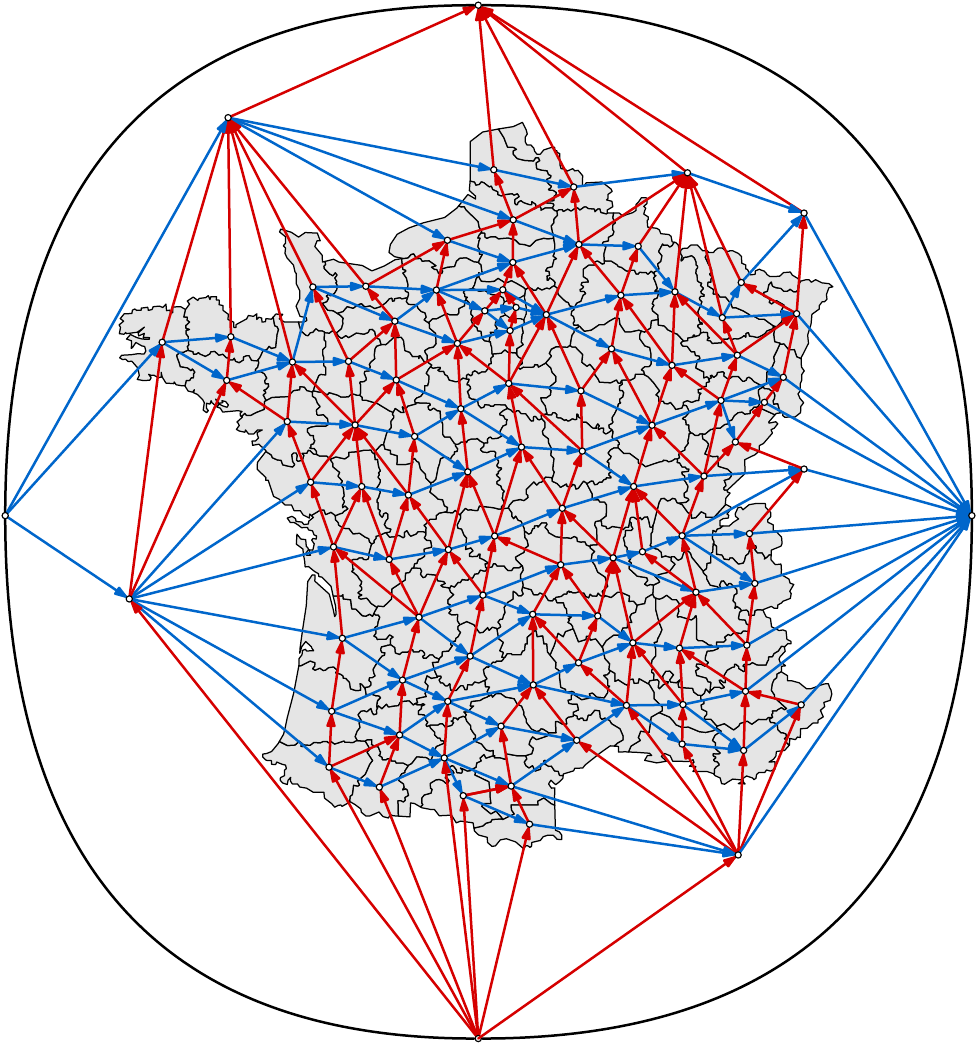}
    \caption{Reference layout for the departments of France dataset. }
    \label{fig:france-transversal}
\end{figure*}

\begin{figure*}[p]    
    \centering
    \includegraphics[page=1]{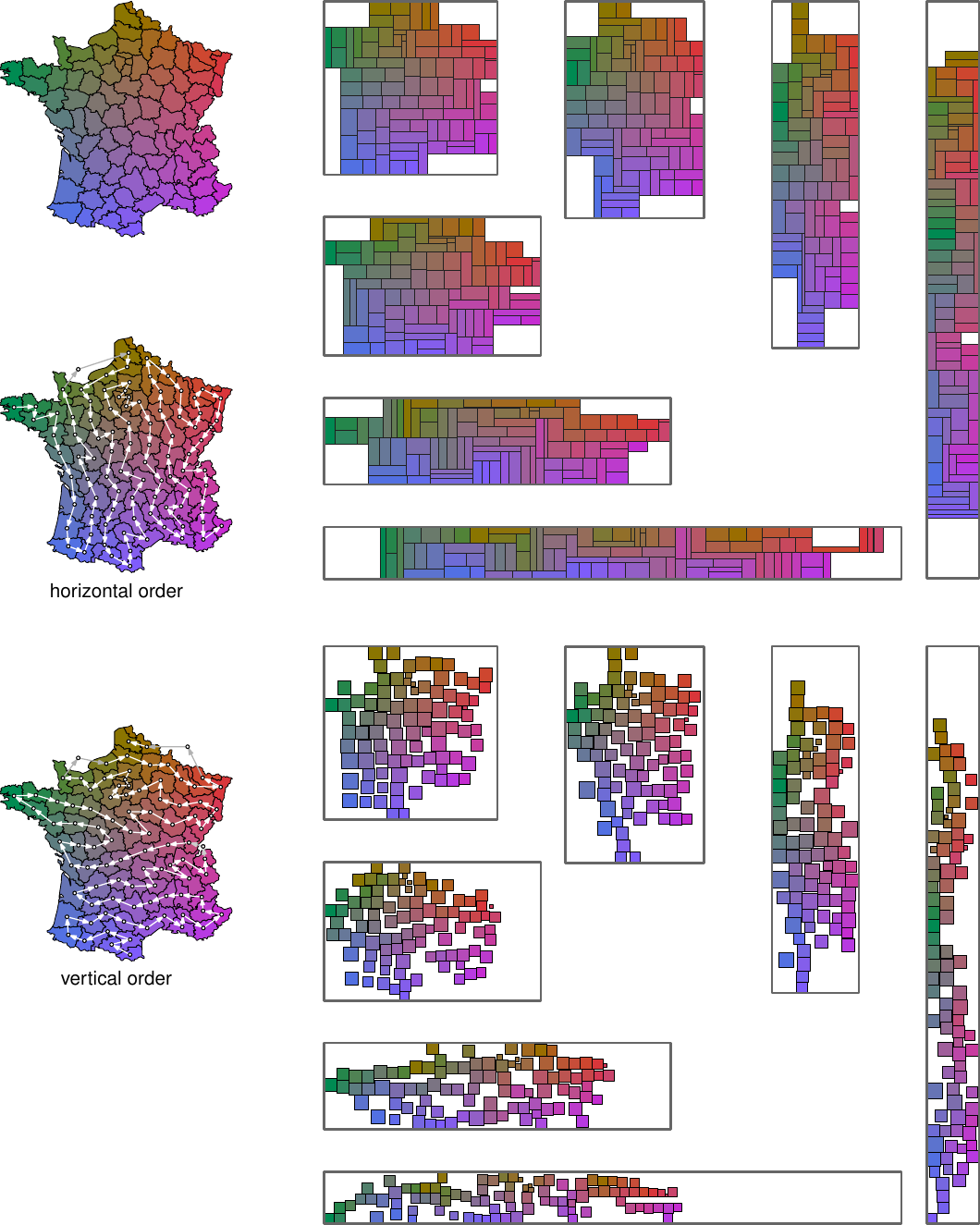}
    \caption{Results for the min-width heuristic on the departments of France. (left) Geographic map and its extremal horizontal and vertical orders of the land regions. (top) Rectangular cartograms produced using layout guides from the map arranger. (bottom) Demers cartograms produced using layout guides from the map arranger.}
    \label{fig:france_min-width_output}
\end{figure*}

\begin{figure*}[p]
    \centering
    \includegraphics[page=2]{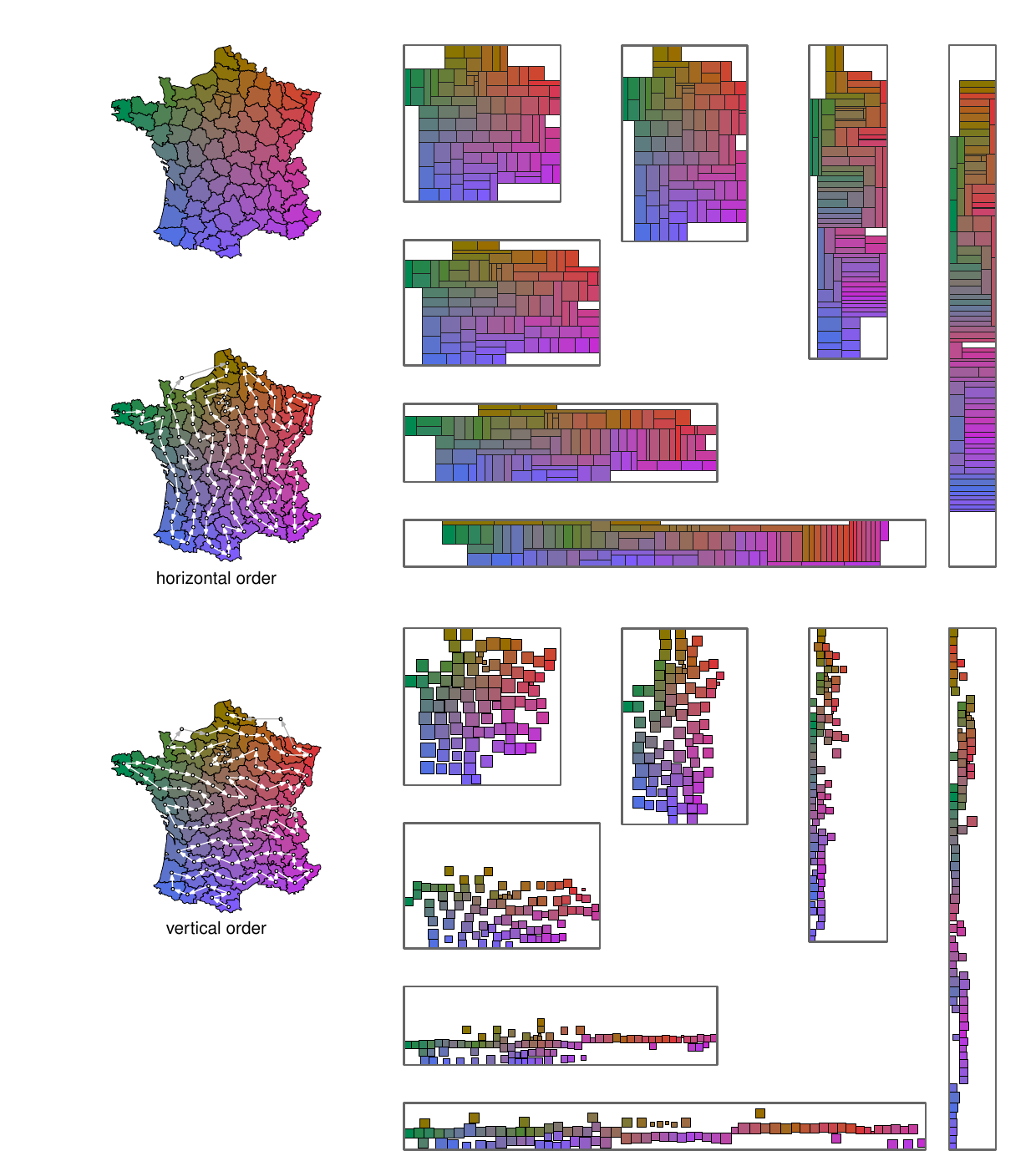}
    \caption{Results for the min-change heuristic on the departments of France. (left) Geographic map and its extremal horizontal and vertical orders of the land regions. (top) Rectangular cartograms produced using layout guides from the map arranger. (bottom) Demers cartograms produced using layout guides from the map arranger.}
    \label{fig:france_min-change_output}
\end{figure*}

\begin{figure*}
    \centering
    \includegraphics[page=1]{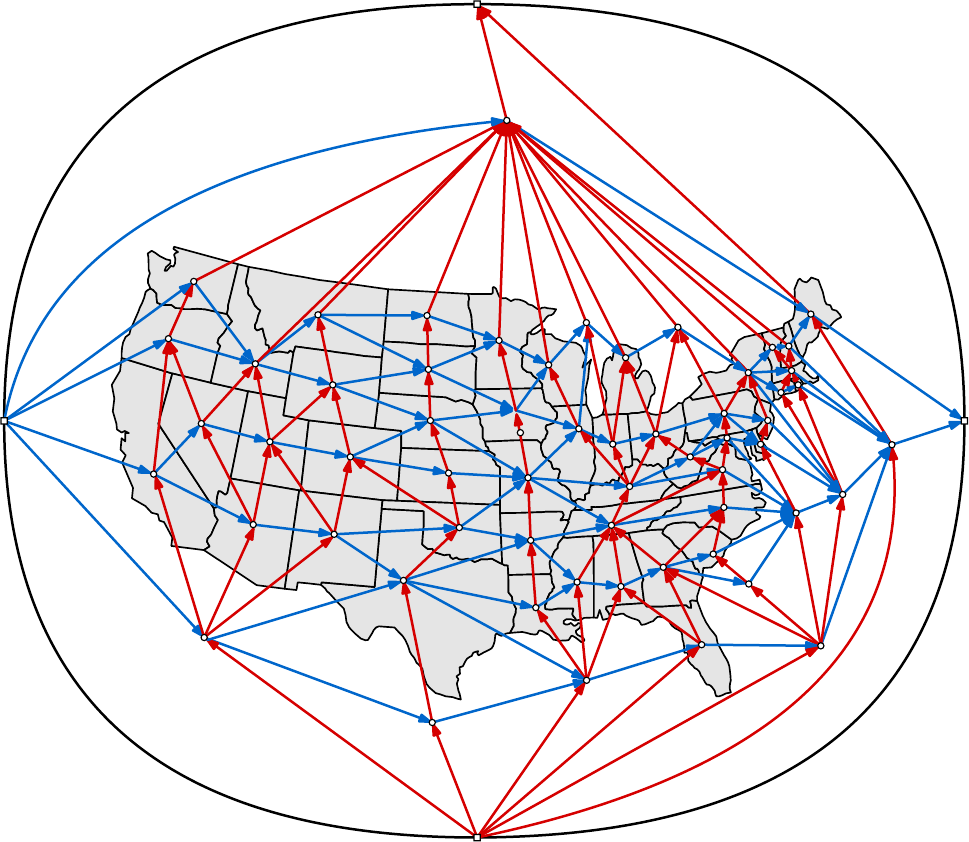}
    \caption{Reference layout for the contiguous US dataset. }
    \label{fig:USA-transversal}
\end{figure*}

\begin{figure*}[p]
    \centering
    \includegraphics[page=1]{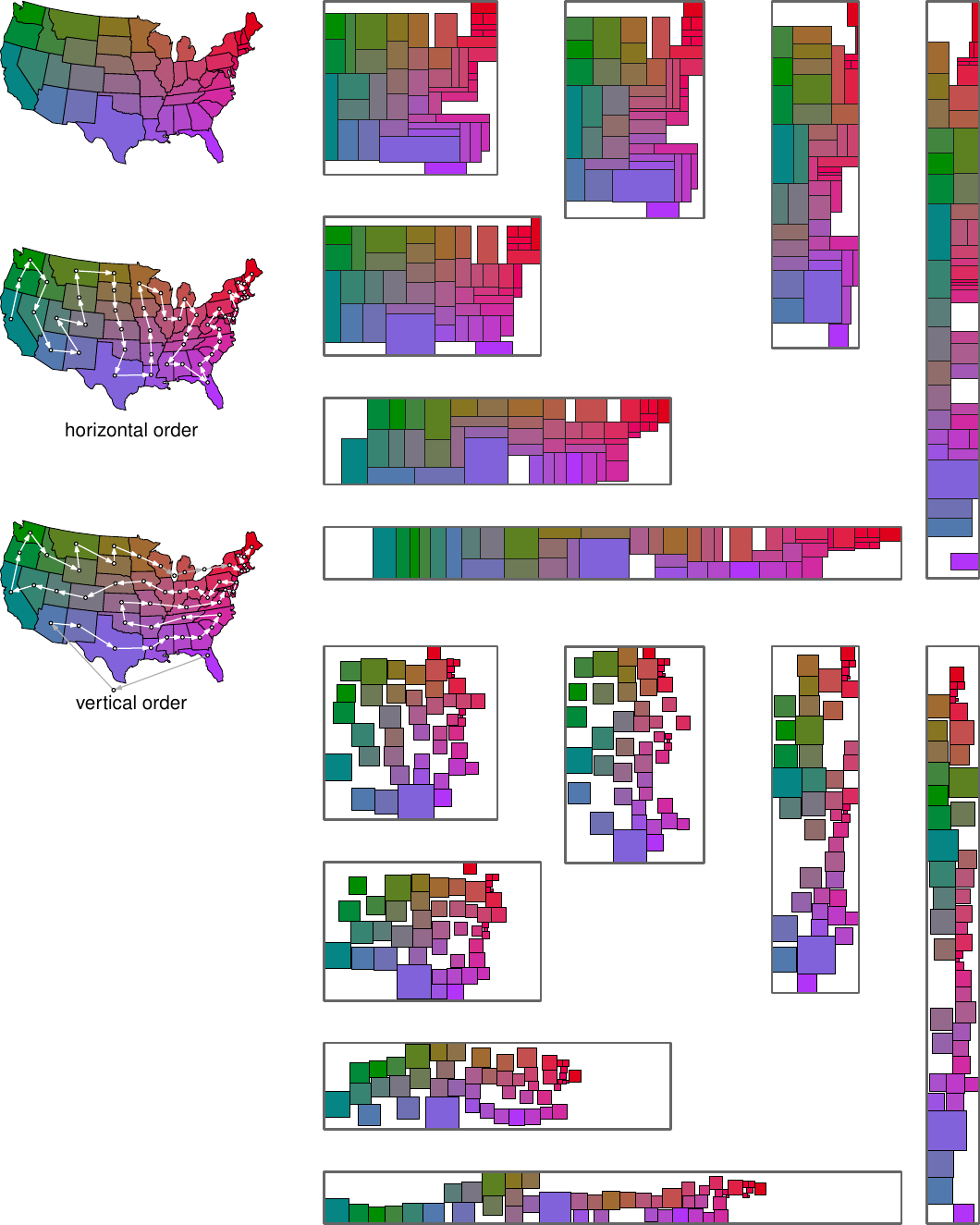}
    \caption{Results for the min-width heuristic on the contiguous US. (left) Geographic map and its extremal horizontal and vertical orders of the land regions. (top) Rectangular cartograms produced using layout guides from the map arranger. (bottom) Demers cartograms produced using layout guides from the map arranger.}
    \label{fig:usa_min-width_output}
\end{figure*}

\begin{figure*}[p]
    \centering
    \includegraphics[page=2]{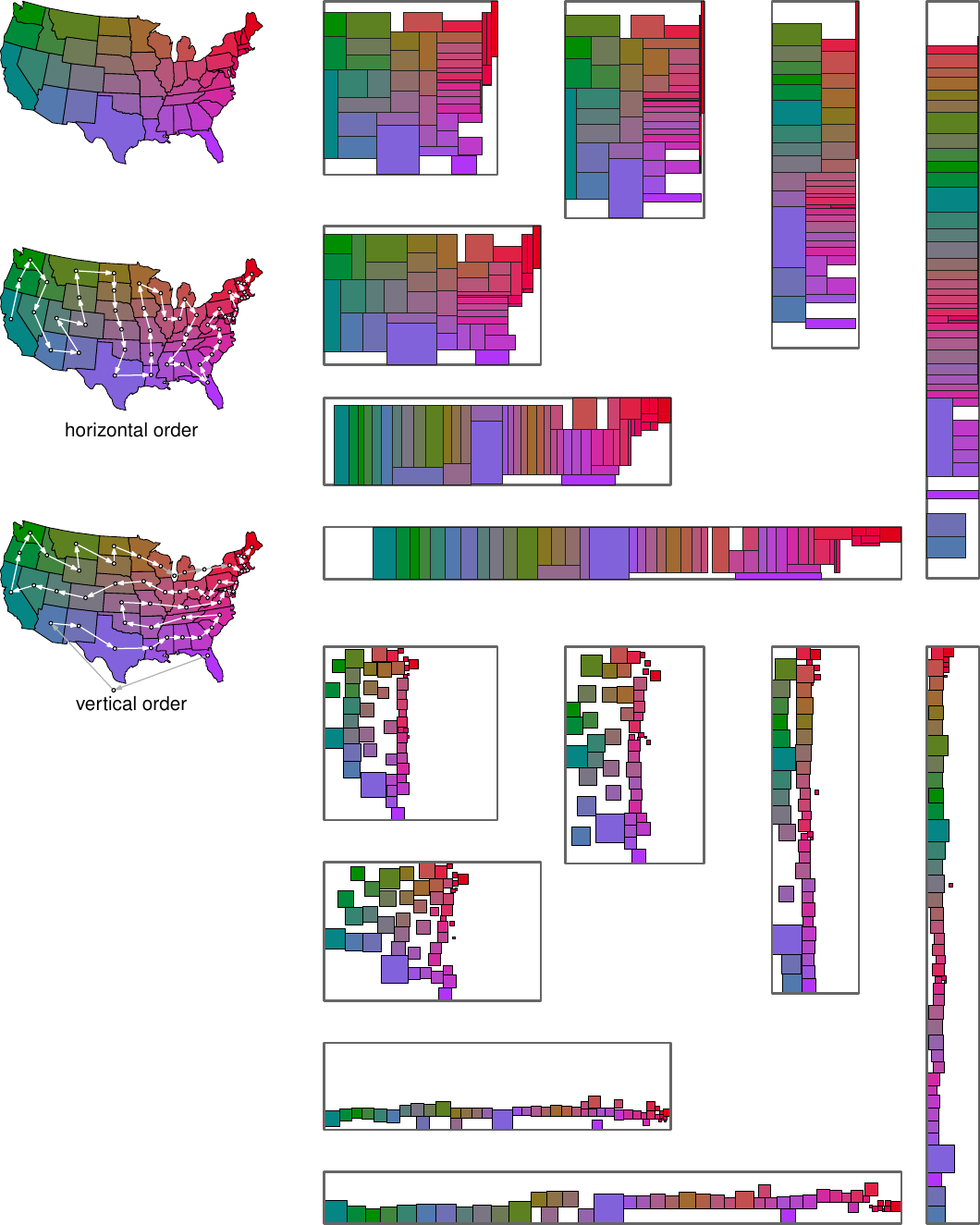}
    \caption{Results for the min-change heuristic on the contiguous US. (left) Geographic map and its extremal horizontal and vertical orders of the land regions. (top) Rectangular cartograms produced using layout guides from the map arranger. (bottom) Demers cartograms produced using layout guides from the map arranger.}
    \label{fig:usa_min-change_output}
\end{figure*}

\end{document}